\documentclass[sigconf]{acmart}

\usepackage[linesnumbered,ruled,vlined]{algorithm2e}
\usepackage{amsthm}
\usepackage[inline]{enumitem}
\usepackage{makecell}
\usepackage{multirow}
\usepackage{subcaption}
\usepackage{xcolor}

\newtheorem{definition}{Definition}

\SetKwComment{Comment}{$\triangleright$\ }{}

\SetCommentSty{mycommfont}

\AtBeginDocument{%
  }

\setcopyright{acmlicensed}
\copyrightyear{2026}
\acmYear{2026}
\acmDOI{10.1145/3769843}
\acmConference[SIGMOD '26]{the ACM SIGMOD/PODS International Conference on Management of Data}{May 31--June 5, 2026}{Bengaluru, India}
\acmISBN{978-1-4503-XXXX-X/18/06}

\begin{document}

\title[A Window-to-Window Incremental Index for Range-Filtering Approximate Nearest Neighbor Search]{WoW: A Window-to-Window Incremental Index for Range-Filtering Approximate Nearest Neighbor Search}

\author{Ziqi Wang}
\orcid{0009-0003-2027-3524}
\affiliation{%
    \department{State Key Laboratory for Novel Software Technology}
    \institution{Nanjing University \country{China}}
}
\email{ziqiw.nju@gmail.com}

\author{Jingzhe Zhang}
\orcid{0009-0001-7062-9096}
\affiliation{%
    \department{State Key Laboratory for Novel Software Technology}
    \institution{Nanjing University \country{China}} 
}
\email{jzzhang.nju@gmail.com}

\author{Wei Hu}
\orcid{0000-0003-3635-6335}
\authornote{Corresponding author}
\affiliation{
    \department{State Key Laboratory for Novel Software Technology}
    \department{National Institute of Healthcare\\ Data Science}
    \institution{Nanjing University \country{China}}
}
\email{whu@nju.edu.cn}


\begin{abstract}
Given a hybrid dataset where every data object consists of a vector and an attribute value, for each query with a target vector and a range filter, range-filtering approximate nearest neighbor search (RFANNS) aims to retrieve the most similar vectors from the dataset and the corresponding attribute values fall in the query range.
It is a fundamental function in vector database management systems and intelligent systems with embedding abilities.
Dedicated indices for RFANNS accelerate query speed with an acceptable accuracy loss on nearest neighbors.
However, they are still facing the challenges to be constructed incrementally and generalized to achieve superior query performance for arbitrary range filters.
In this paper, we introduce a window graph-based RFANNS index.
For incremental construction, we propose an insertion algorithm to add new vector-attribute pairs into hierarchical window graphs with varying window size.
To handle arbitrary range filters, we optimize relevant window search for attribute filter checks and vector distance computations by range selectivity.
Extensive experiments on real-world datasets show that for index construction, the indexing time is on par with the most building-efficient index, and 4.9$\times$ faster than the most query-efficient index with 0.4--0.5$\times$ smaller size;
For RFANNS query, it is 4$\times$ faster than the most efficient incremental index, and matches the performance of the best statically-built index.
\end{abstract}



\maketitle

\section{Introduction} \label{sec:intro}

Nearest neighbor search (NNS) \cite{knn-problem} is widely used as a fundamental process in various applications.
On e-commerce platforms, people use an image to search for similar commodities \cite{vearch-vdb,milvus-vdb,adbv-vdb,vdb-survey,dse-kge}.
In the retrieval-augmented generation (RAG) pipeline \cite{rag-survey,corag}, people retrieve paragraphs relevant to a question and provide them as context for large language models.
These functions can be realized by encoding the original data into vectors with modern embedding models \cite{mteb,clip,transe,glove} and persisting with a vector database management system (VDBMS) \cite{manu-vdb,vdb-survey}.
With a query also embedded into the same vector space, VDBMS uses $k$-NNS to find similar results in the database.
However, finding the exact $k$ nearest neighbors is time-consuming because the time complexity is $O(nd)$, where $n$ is the number of vectors and $d$ is the dimension of vectors \cite{efficient-dco,finger-dco}.
On real-world datasets, the query latency can be even unacceptable for online scenarios, as $n$ varies from millions to billions while $d$ is in the order of hundreds or thousands.

Recently, a family of algorithms are proposed without having to retrieve all exact nearest neighbors in trade for sublinear query time complexity \cite{ann-survey,vdb-survey,billion-scale-chenhaibo,spann,curse_of_d}. 
It is referred to as \emph{approximate nearest neighbor search} (ANNS) and there exist mainly three categories: hashing-based \cite{lsh,l2h,db-lsh2,det-lsh,lsh-lcc}, partition-based \cite{pq,pqfs,local-pq}, and graph-based algorithms \cite{graph-survey, nsw, hnsw, tao-mg, nsg, nssg, fast-knng}.
With the rapid change in the demand for various query scenarios, it is inadequate to retrieve data only by vector similarity \cite{milvus-vdb,vdb-survey}.
A popular scenario is called \emph{filtering search} \cite{filtered-diskann,hqann,nhq-nips23,acorn}, which aims to find the most similar vectors whose attribute payloads can pass a certain predicate filter.
In this work, we focus on the case where the dataset has only one attribute and the predicate is a range (a.k.a. window) filter, which is named \emph{range-filtering ANNS} (RFANNS) \cite{arkgraph,serf,wst,irange,hsig,vdb-survey}.

For instance, instead of searching for similar commodities solely by an image, we can provide an additional price range to retrieve affordable ones.
Another example resides in a medical question answering system, a typical RAG application~\cite{medical-rag,dataintelli-clinical}.
When a query comes with \textit{``What symptoms are common for people with hypertension aged from 50 to 60?''}, the retriever may generate an RFANNS query to find medical records about \textit{``hypertension''} tagged with ages of patients between 50 and 60.

Generally, there are three methods to address the RFANNS problem \cite{serf,irange,wst}.
\emph{Pre-filtering} first selects in-range vectors from the dataset.
The term ``in-range'' indicates the attribute value alongside the vector can pass the range filter while the term ``out-of-range'' is opposite, hereafter.
As there is no index built for in-range vectors, it can be inefficient to search by linear scan.
\emph{Post-filtering} builds an ANNS index over the entire dataset, and retrieves some intermediate vectors before eliminating the out-of-range vectors.
It may cause the number of query results less than $k$, and thus another trial to retrieve more intermediate vectors is needed.
\emph{Dedicated RFANNS indices} only visit in-range vectors during searching, called \emph{in-filtering}.
For partition-based indices \cite{rangepq}, filtering happens in cluster discovery and posting list scan.
Graph-based indices \cite{arkgraph,serf,wst,irange,hsig,digra} aim to simulate the proximity graph exclusively built over in-range vectors, referred to as \emph{oracle proximity graph} \cite{acorn,irange}.
Figure~\ref{fig:example}a shows an oracle proximity graph with $\text{outdegree}=2$ built over vectors in the range $[5,15]$.
It has been acknowledged that graph-based RFANNS indices are superior in query efficiency over non-graph-based solutions \cite{vdb-survey,serf,rangepq}.

\begin{figure}[t]
    \centering
    \includegraphics[width=.74\linewidth]{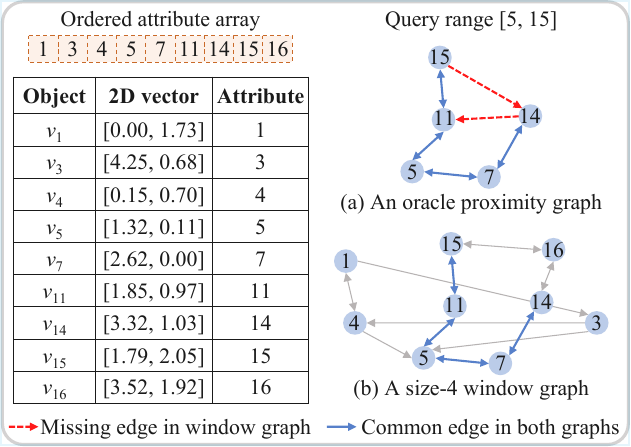}
    \caption{Example of (a) an oracle proximity graph and (b) a size-4 window graph. Dataset is on the left.}
    \Description{Example}
    \label{fig:example}
\end{figure}

In analyzing applications where dedicated RFANNS indices are used, we find two key challenges to index building and searching:

\textbf{Challenge 1.} \textit{Can RFANNS indices be built incrementally?}
Dedicated graph-based RFANNS indices are slower to build compared with hashing/partition-based ones \cite{ann-survey,revisiting-graph-building,rangepq}.
Also, as materializing all proximity graphs for arbitrary range filters is unrealistic, they attempt to build a denser graph that covers edges of all oracle graphs, at the cost of longer indexing time and larger size than a vanilla graph built on the vector set.
If the index is static \cite{wst,irange} or suffers query performance degradation due to limited increment support\cite{rangepq,digra,serf}, it must be rebuilt upon new data arrivals or maintained via a small auxiliary index before periodic merging.
Considering the high reconstruction cost, an incremental index is preferable to skip frequent rebuilding or merging  \cite{fresh-diskann,hsig}.
However, most RFANNS indices are either static or only provide limited incremental support \cite{serf,wst,irange,rangepq, digra}, presenting an opportunity to devise an index with unrestricted incremental capabilities.

\textbf{Challenge 2.} 
\textit{Can RFANNS indices handle varying query correlations and selectivity}?
Another challenge is the no or negative correlation \cite{acorn} between the query vector neighborhood and the filtering subset \cite{acorn,irange,ung}.
For example, the most similar commodities to an image may not satisfy the requested price range.
The performance on such workloads may largely affect the overall efficiency.
Queries of high selectivity are also challenging as most of nearest neighbors are out of range.
Many RFANNS indices \cite{serf,wst,hsig,rangepq} fail to obtain the optimal balance between query speed and accuracy for all levels of query correlations and selectivity.
As far as we know, only iRangeGraph \cite{irange} and DIGRA \cite{digra} can achieve close performance to the oracle graph.
However, iRangeGraph is a static index and lacks support for insertion and DIGRA encounters severe accuracy loss after frequent insertions.
There is still room for improvement in index building and searching by more advanced structures and optimized query algorithms.

In this paper, we design a \emph{window graph}-based RFANNS index.
Intuitively, if two vertices are connected in the oracle proximity graph of a query range, they should be similar in both attribute values and vectors.
Through an ordered array, the attribute similarity for a given value can be modeled using an attribute window that is centered at the value and extended symmetrically in two directions.
Considering attribute and vector similarity, a vertex in a window graph only connects to in-window nearest vertices.
When the range cardinality matches the window size, the window graph achieves optimal in both similarity metrics.
Figure~\ref{fig:example}b shows a size-4 window graph on the ordered attribute array. 
Each vertex connects to its two nearest neighbors, in accord with the oracle graph's outdegree constraint.
For example, $v_7$ connects to $v_5$ and $v_{14}$ as they are more similar in its window $[4,5,11,14]$ than $v_4$ and $v_{11}$.

However, only one window graph cannot generalize to varying query cardinality, leading to candidate missing due to out-of-range vertices.
As shown in Figure~\ref{fig:example}, the window graph misses 2 of 10 directed edges in the oracle graph.
To recover the missing connections, smaller window graphs can be built as complements.
For example, in a size-2 window graph, the missing edges can be retrieved from the neighborhoods of $v_{14}$ and $v_{15}$, whose windows are $[11,15]$ and $[14,16]$, respectively.
Large and small window graphs together form a hierarchy for RFANNS indexing and searching.

Based on the above intuition, we propose \textbf{WoW}. 
Its basic idea is to map multiple \textbf{W}indow graphs t\textbf{o} \textbf{W}indow filters.
To address Challenge 1, we design an efficient algorithm that supports inserting vector-attribute pairs into hierarchical window graphs with varying window size.
To resolve Challenge 2, we optimize traversal on window graphs of different size to tightly cover arbitrary range filters with varying query correlations.
We conduct extensive experiments on benchmark datasets to evaluate WoW and its key components.
The experimental results show the superiority of WoW in index building and query efficiency against state-of-the-art dedicated RFANNS indices.
The code and experimental data can be visited at \url{https://github.com/nju-websoft/WoW}.

The main contributions of this paper are outlined as follows:
\begin{itemize}
    \item 
    For index building, WoW is fully incremental from an empty index without data presorting or partial indexing.
    The worst-case time complexity of insertion scales to $O(\log^2n)$, which can be highly parallelized to further boost building speed.
    
    \item 
    For RFANNS query, we employ query selectivity to choose the most relevant window graphs. As far as we know, we are the first to reduce filter tests of attribute values and distance computations of vectors at the same time.
    We prove that the query time complexity is $O(\log n')$, equivalent to that on the oracle relative neighbor graph (RNG) with $n'$ in-range vectors.
    We also conduct a theoretical analysis to estimate the optimal parameter setting for best query performance.
    
    \item 
    Extensive experiments show the efficiency of WoW:
    \begin{enumerate*}
    \item The indexing time of WoW competes with the most building-efficient baseline \cite{serf}, and is 4.9$\times$ faster with 0.4–0.5$\times$ smaller size than the most query-efficient index \cite{digra}.
    \item WoW delivers 4$\times$ faster queries than the best incremental index~\cite{hsig} across all workloads.
    \item WoW is 1.5$\times$ faster than the best statically-built index \cite{digra} on workloads of high selectivity, with equal or better performance on others.
    \end{enumerate*}
\end{itemize}

\begin{table}
\centering
\caption{Frequently-used notations}
\label{tab:notation}
{\small
\begin{tabular}{lp{6cm}}
\toprule
    Notation & Description\\
\midrule
    $\mathcal{D} =\{\mathcal{V},\mathcal{A}\}$ & Hybrid dataset $\mathcal{D}$ consisting of vector dataset $\mathcal{V}$ and attribute set $\mathcal{A}$, s.t. $|\mathcal{V}|=|\mathcal{A}|=n$\\
    $v_a$         & Vector $v$ with attribute value $a$\\
    $R=[x,y]$     & Query range $R$ with left and right boundaries\\
    $L=[l_{\min},l_{\max}]$ & Layer range $L$ with lower and upper limits\\
    $N_v^l$       & Neighbor set of $v$ at layer $l$\\
    $W_a^l$       & Window of $a$ at layer $l$\\
    $o$           & Window boosting base\\  
    $\omega_c$    & Beam search width in index construction\\   
    $m$           & Graph maximum outdegree\\
\bottomrule
\end{tabular}}
\end{table}

\section{Preliminaries} \label{sec:prelim}
In this section, we formulate the studied problem, followed by a survey of related work.
Table~\ref{tab:notation} lists the frequently used notations.

\subsection{Problem Formulation} \label{sec:intro-problem_def}

\begin{definition}[Approximate Nearest Neighbor Search]
    A vector dataset $\mathcal{V}$ is defined in vector space with a distance metric $\delta$, where each vector has dimension $d$.
    Given a query vector $q$, ANNS aims to find a subset $\mathcal{S}\subseteq\mathcal{V}$ with $k$ vectors to minimize $\sum_{v\in \mathcal{S}} \delta(v,q)$.
\end{definition}

Oftentimes, the result contains non-closest vectors.
The search accuracy can be defined by the fraction of the exact nearest neighbors, i.e. $\textit{Recall}=\frac{|\mathcal{S}\cap \mathcal{G}|}{k}$, where $\mathcal{G}$ is the set of the exact nearest neighbors (i.e. gold standard).

RFANNS aims to find vector-attribute pairs such that vectors are the closest to the query vector and attribute values are in range.
\begin{definition}[Range-Filtering ANNS]
    A hybrid dataset $\mathcal{D} = \{\mathcal{V}, \mathcal{A}\}$ consists of a vector set $\mathcal{V}$ and an attribute set $\mathcal{A}$ with equal size $n$.
    Given a range-filtering query $(q, R)$, where $q$ is the query vector and $R = [x, y]$ is the range (window) filter over $\mathcal{A}$, with $x$ and $y$ as its left and right boundaries, RFANNS aims to find a subset $\mathcal{S}\subseteq\mathcal{D}$ with $k$ vector-attribute pairs to minimize $\sum_{(v,a)\in \mathcal{S}} \delta(v,q), x\le a \le y$.
    We call a pair $(v,a)$ in-range only when $x\le a\le y$ (i.e. $a\in R$).
\end{definition}

The selectivity of a query range indicates how many attribute values are in range.
High selectivity means that only a few in the attribute set can pass the range filter.
\begin{definition} [Query Selectivity] \label{def:selectivity}
    For a range filter $R=[x,y]$, let $n'$ denote the number of attributes in $\mathcal{A}$ that are in range $R$,
    the fraction of in-range attributes over the dataset cardinality is $f=\frac{n'}{n}$, and the selectivity of $R$ is defined by $s=\frac{1}{f}$.
\end{definition}

Note that if $n'$ is less than $k$ in the entire dataset, recall should be calculated by $\textit{Recall}=\frac{|\mathcal{S}\cap\mathcal{G}|}{n'}$.

\subsection{Related Work} \label{sec:prelim-related}

\paragraph{ANNS indices}
Existing works \cite{ann-survey, vdb-survey} can be categorized into hashing-based, partition-based, and graph-based.
Hashing-based indices adopt static \cite{lsh,det-lsh} or learnable \cite{l2h,db-lsh2} hash functions to map similar vectors to identical or nearby hash buckets. 
Partition-based indices \cite{tree-based,pq, pqfs,rabitq,optimized-pq} group similar vectors in the same cluster. 
Graph-based ones \cite{nsw,hnsw,tao-mg,nsg, nssg,fast-knng,symphonyqg} represent vector similarity by edges, aiming to approximate the RNG: vertices $r,s\in V$ are connected if and only if $\forall t\in V, \delta(r,s)<\delta(r,t)$ or $\delta(r,s)<\delta(s,t)$ \cite{graph-survey,graph-survey-2}.
Hierarchical navigable small world (HNSW) \cite{hnsw} is widely used for fast indexing and superior query performance.

Graph-based indices are further divided into refinement-based and increment-based by construction strategies \cite{revisiting-graph-building}.
Increment‐based ones (e.g., NSW \cite{nsw} and HNSW \cite{hnsw}) offer greater flexibility than static, refinement‐based ones (e.g., NSG \cite{nsg} and NSSG \cite{nssg}).
We focus on the challenge of incremental support, which remains insufficiently addressed yet.
Deletion for graph‐based indices is addressed by a separate line of works \cite{fresh-diskann,topology-repair,gti}.
They exploit the RNG property to reconstruct the neighborhood around a deleted vertex, making them applicable to any underlying graph structure.

\paragraph{Filtering ANNS and RFANNS}
There are mainly three types of filters: label filter, range filter, and generic predicate filter.
Some indices \cite{filtered-diskann,nhq-nips23,ung,hqann,rwalks} attempt to resolve label-filtering ANNS.
For generic query predicate, VDBMS and ANNS libraries \cite{adbv-vdb,vearch-vdb,milvus-vdb,pase-vdb,vbase-vdb,weaviate-vdb,single-store,faiss,chase-vdb,parlay-ann} estimate query selectivity to pick pre-filtering or post-filtering.
As the inherent defects of pre/post-filtering, dedicated indices \cite{acorn,hqi,filter-ivf} are designed to only visit those vectors whose corresponding attributes can satisfy the predicate.

\begin{table}
    \centering
    \caption{Comparison of RFANNS indices}
    {\small
    \begin{tabular}{lcccc}
    \toprule
       Algorithm & ANNS & Indexing & OOR & Key structure\\
    \midrule
       Pre-filtering & - & Increment & $\times$ & -\\
       Post-filtering & Any & Increment & $\checkmark$ & -\\
       SeRF & HNSW & Ordered inc. & $\times$ & Segment graph\\
       WST & Vamana & Static & $\checkmark$ & Segment tree\\
       iRangeGraph & NSW & Static & $\times$ & Segment tree\\
       RanePQ & PQ & Post-increment & $\times$ & WBT\\
       DIGRA & NSW & Post-increment & $\times$ & Multi-way tree\\
       HSIG & HNSW & Increment & $\checkmark$ & Unified graph\\
       WoW (ours) & NSW & Increment & $\times$ & Window graph\\
    \bottomrule
    \end{tabular}}
    \label{tab:rfann-comp}
\end{table}

Graph-based RFANNS indices aim to search on proximity graphs built solely over in-range vectors.
SeRF \cite{serf} attempts to compress oracle HNSWs for all ranges.
But the compression is not lossless, and some less proximate edges are retained, which may harm query speed.
WST \cite{wst} and iRangeGraph \cite{irange} are based on segment tree \cite{segment-tree}.
A single graph-based index is built for each tree node.
WST runs several ANNS on tree nodes whose range intersects with the query filter and merges separate results.
It inherits the weakness of post-filtering that the distance of out-of-range (OOR) vectors may be calculated.
iRangeGraph \cite{irange} acquires candidates by tree traversal at each hop.
However, tree traversal consumes $O(\log n)$ time overhead that may diminish query speed, especially for filters of high selectivity and low-dimensional datasets \cite{irange}. 
All RFANNS indices above are static indices \cite{wst,irange} or only support attribute-ordered insertions \cite{serf}.
They need to presort the vectors based on attribute values before construction.
RangePQ \cite{rangepq} and DIGRA \cite{digra} support insertion after static construction on a subset, namely \emph{post-incremental} indices.
RangePQ is the only partition-based RFANNS index.
It uses product quantization (PQ) to partition the subset and a weighted balanced tree (WBT) \cite{wbt} to group coarse cluster centers within the same range.
However, clusters have to be retrained to handle distribution shift due to insertions.
DIGRA presorts the subset like static indices and recursively builds a multi-way tree structure, before building separate HNSW inside each tree node.
For insertion, it makes B-tree-like splitting on tree nodes but may break graph interconnections.
HSIG \cite{hsig} supports unordered insertion.
It designs a learnable estimator to navigate queries with different selectivity to pre-filtering, post-filtering, and a dedicated index.

Table~\ref{tab:rfann-comp} compares the features of representative RFANNS indices:
\begin{enumerate*}
\item \emph{ANNS} indicates the ANNS algorithm that an index is based on.
\item \emph{Indexing} compares how to construct the indices: static, post-incremental, ordered incremental, or unordered incremental.
\item \emph{OOR} marks whether the distance of out-of-range vectors needs be calculated.
\item \emph{Key structure} lists the structure employed to arrange vectors and attribute values.
\end{enumerate*}
This table suggests that WoW is the only index to support fully incremental construction (cf. Challenge~1) without OOR-vector visit (cf. Challenge 2).

\section{Window-to-Window Incremental Index}
In this section, we first give the definition of \emph{hierarchical window graphs} and the structure of WoW.
Then, we describe how to incrementally build the index and use it to answer RFANNS queries.
For illustration simplicity, we assume attribute values in $\mathcal{A}$ are all \emph{unique}.
However, WoW is \emph{not} limited by this assumption, and we show how to handle duplicate values at the end of this section.

\subsection{WBT and Hierarchical Window Graphs}

\renewcommand{\thedefinition}{4}
    \begin{definition}[Window Graph] \label{def:window_graph}
    A window graph is a directed graph $G=\{\mathcal{P},\mathcal{E},w\}$, where $\mathcal{P}$ is the vertex set, $\mathcal{E}$ is the edge set, and $w$ is the half window size.
    Each vertex $v_i$ denotes a vector-attribute pair, where $v$ is the vector and $i$ is the attribute value. 
    $\mathcal{E}$ satisfies:

    (1) RNG property:
        $\forall(v_i,v'_j)\in \mathcal{E}, v''_k\in\mathcal{P}\backslash\{v_i,v'_j\}$, then $\delta(v_i,v'_j) < \delta(v_i,v''_k) \vee \delta(v_i,v'_j)<\delta(v'_j,v''_k)$;
        
    (2) Window property:
    $\forall(v_i,v'_j)\in \mathcal{E}$, then $|rank(i)-rank(j)| < w$, where $rank(i)$ counts the number of unique values less than $i$. 
\end{definition}
\renewcommand{\thedefinition}{5}
\begin{definition}[Hierarchical Window Graphs] \label{def:hwg}
    Hierarchical window graphs are a set of window graphs $\mathcal{H}=\bigcup_{l\in[0,\textit{top}]}G_l$.
    Define a window boosting base $o\geq 2$, then $G_l.w=o^l$ and $\textit{top}=\lceil\log_o\frac{|\mathcal{A}|_u}{2}\rceil$, where $|\mathcal{A}|_u$ is the total number of unique values.
\end{definition}
\renewcommand{\thedefinition}{\arabic{definition}}

Figure~\ref{fig:struct} depicts an example of WoW over a hybrid 2D dataset, with maximum outdegree $m=4$ and window boosting base $o=4$.
Numbers on tree nodes and graph vertices represent the attribute values.
The distance of vertices in the graph corresponds to the distance of vectors, and edges are all directed.

On the left side, a \emph{weighted balanced tree} (WBT) \cite{wbt} is used to calculate the windows of an attribute in different layers.
Every node in WBT records the rooted tree size. 
WoW can accurately maintain windows of vertices and easily retrieve in-window neighbors.
An ordered array (Figure~\ref{fig:example}) can also be used to compute windows. 
But it suffers a degradation to linear insertion complexity, whereas WBT achieves logarithmic time.
Unlike RangePQ \cite{rangepq} tightly coupling the WBT and PQ structures into an integrated framework, WoW adopts WBT as a lightweight plug-in to accelerate attribute insertion, replacing the ordered array.

On the right side, there are three \emph{window graphs} from Layer~0 to Layer~2 (the \textit{top} layer).
In a window graph, each vertex is only aware of the proximate vertices in a fixed-size window of the attribute.
It can handle RFANNS queries where the range filter is compatible with the window size.
To handle arbitrary range filters, we build \emph{hierarchical window graphs} with varying size.
For visual clarity, in our example, only edges of $v_{74}$ in Layer 0 and Layer 1 are shown, and it has two out-edges and two in-edges in Layer~0.
Neighbors and windows of some vertices are also provided on the right.
Compared with the multi-layer structure of HNSW \cite{hnsw} for faster vector neighborhood approaching, WoW deploys hierarchy to store neighbors in varying-sized windows, designed to serve the attribute instead of the vector.

\begin{figure}[t]
    \centering
    \includegraphics[width=\linewidth]{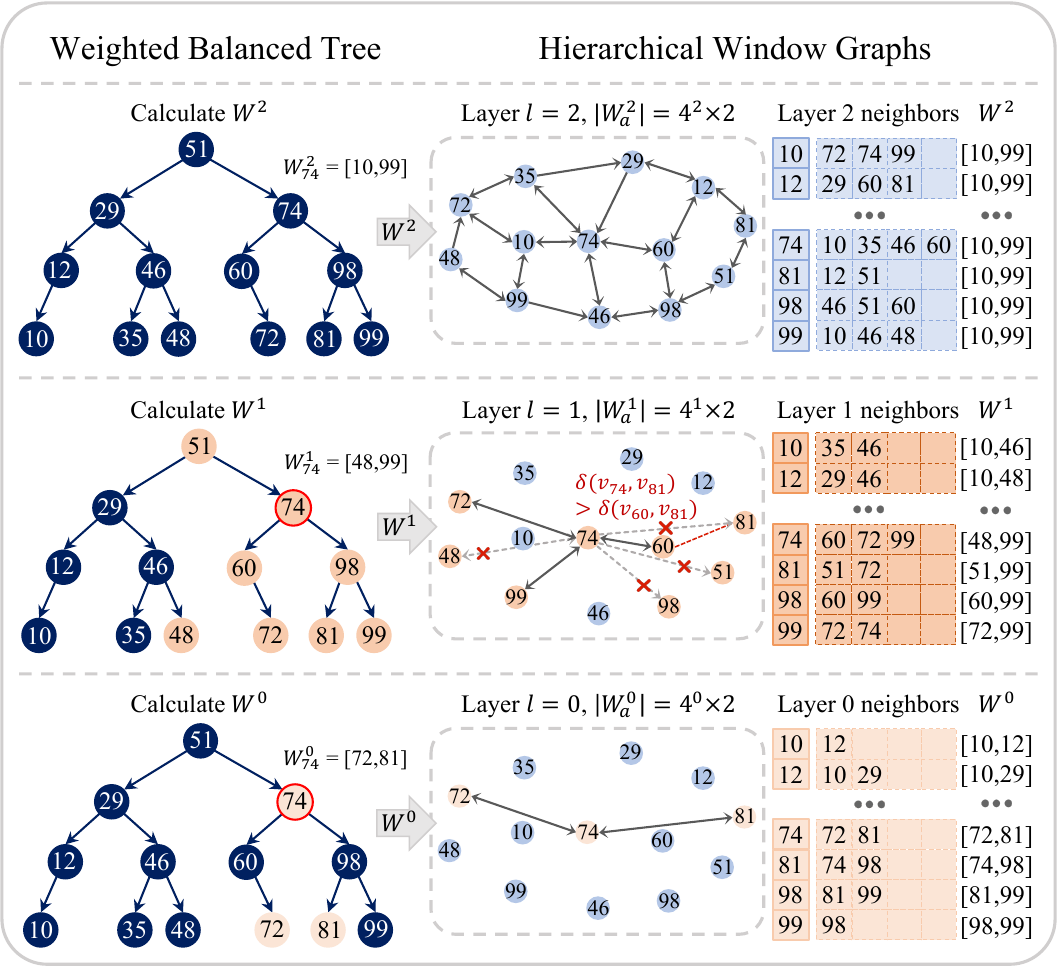}
    \caption{Index structure of WoW}
    \Description{Index structure of WoW}
    \label{fig:struct}
\end{figure}

The window of an attribute value $a$ contains an equal number of ordered values halved by $a$, where the number is calculated by $o^l$.
If there are inadequate values inside the window, which may happen at the boundaries of the ordered dataset, the window boundaries would be limited to the dataset boundaries.
For example in Figure~\ref{fig:struct}, $W_{74}^1=[48,99]$ as $o^l=4$ in Layer $l=1$, and the 4th smaller value of 74 is 48, which is used as the left window boundary.
As there are only three values greater than $74$, the right window boundary is set to the right boundary of the dataset, $99$.
The size of top-layer windows (i.e. $|W^2|$) is greater than the number of inserted attribute values ($|\mathcal{A}|$), yielding a proximity graph only aware of vectors.

WoW leverages the \emph{RNG pruning} strategy (abbr. $\mathrm{RNGPrune}$) inside windows to diversify the neighbor distribution and strengthen the navigation ability.
It has been widely studied as one of the most effective way to accelerate query speed \cite{graph-survey-2,deg,hnsw,nsg}.
For example, in Layer 1, $v_{81}$ is not selected as the neighbor of $v_{74}$, because edge $v_{74}\rightarrow v_{81}$ \emph{violates} the RNG property in Definition~\ref{def:window_graph} (i.e. $\delta(v_{74}, v_{81}) > \delta(v_{74}, v_{60})\wedge \delta(v_{74}, v_{81}) > \delta(v_{60}, v_{81})$).
As a result, it becomes the longest edge in triangle $\Delta v_{74}v_{60}v_{81}$ and thus deleted.
The final neighbors of $v_{74}$ are settled as $v_{60}, v_{72}$, and $v_{99}$.

\subsection{Top-down Insertion}

\begin{algorithm}[t]
\SetNoFillComment
\SetKwData{KwHyper}{Hyperparameter}
\SetKwInOut{HyperIn}{Hyperparameter}
\SetKw{KwContinue}{continue}
\caption{Insert}
\label{alg:insert}
\KwIn{$v_a$: vector $v$ with attribute value $a$ to insert}
\KwOut{index $I$ with $v_a$ inserted}
\HyperIn{$m,o,\omega_c$: defined in Table~\ref{tab:notation}}
$\textit{top}\leftarrow$ top layer of $I$\;
\If(\Comment*[f]{top window cannot cover $\mathcal{A}$}){$|\mathcal{A}| + 1 > 2 o^{\textit{top}}$}{
    Clone graph at \textit{top} layer to \textit{top} + 1 layer\;
    $\textit{top}\leftarrow \textit{top} + 1$;}
    \For{$l \leftarrow \textit{top}$ \KwTo 0}{
    $W_a^l\leftarrow$ window of size $2o^l$ halved by $a$\Comment*{request WBT}
    $\textit{ep}\leftarrow$ a random vertex with attribute value in $W_a^l$\;
    $L\leftarrow[l, \textit{top}]$, $U\leftarrow \{v_i\,|\,v_i\in U^{l+1} \wedge i\in W_a^l\}$\Comment*{$U^{\textit{top}+1}=\emptyset$}
    \lIf(\Comment*[f]{skip beam search}){$|U| > m$}{$U^l\leftarrow U$}
    \lElse{$U^l\leftarrow U \cup \mathrm{SearchCandidates} (\textit{ep},v_i,W_a^l,L, \omega_c)$}
    $N_{v_a}^l\leftarrow\mathrm{RNGPrune}(v_a, U^l, \frac{m}{2})$\Comment*{select $\frac{m}{2}$ neighbors}
    \ForEach(\Comment*[f]{adjust neighbors}){$v_b \in N_{v_a}^l$}{
        \If{$|N_{v_b}^l| < m$}{
            $N_{v_b}^l\leftarrow N_{v_b}^l \cup\{v_a\}$, 
            \KwContinue\;
        }
        $W_b^l\leftarrow$ window of size $2o^l$ halved by $b$\;
        $U' \leftarrow {\{v_a\}}\cup\{v_i\,|\,v_i \in N_{v_b}^l \wedge i \in W_b^l \}$\Comment*{pruning 1}
        $N_{v_b}^l\leftarrow\mathrm{RNGPrune}(v_b, U', m)$\Comment*{pruning 2}}
}
Insert $a$ into WBT and connect $\bigcup_{l=0}^{\textit{top}} N_{v_a}^l$ from $v_a$\;
\Return $I$\;
\end{algorithm}

\begin{figure*}
    \centering
    \includegraphics[width=0.99\linewidth]{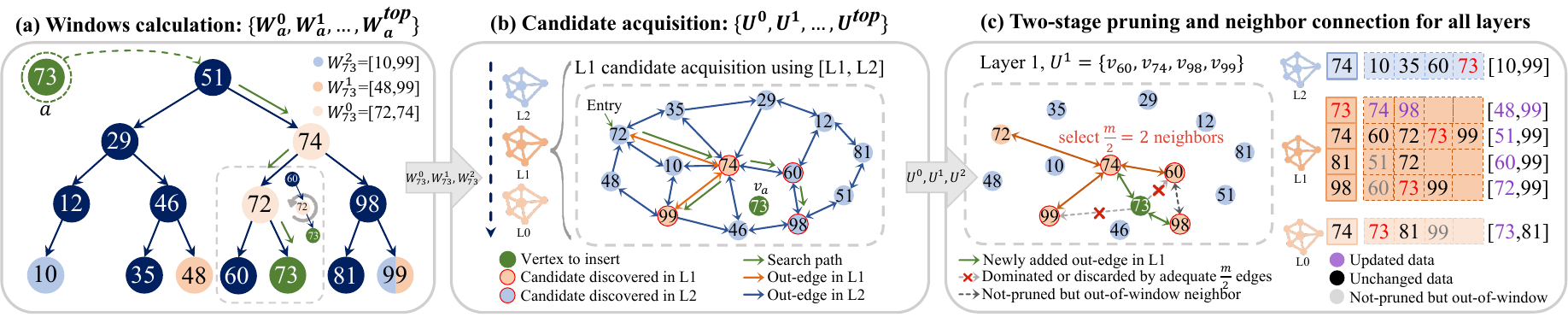}
    \caption{Single insertion of WoW.
    (a) The windows of $v_{73}$ in different layers, and insertion triggers the self-balancing of WBT.
    (b) The procedure of acquiring candidates for Layer 0 to Layer \textit{top}, taking Layer 1 for instance.
    (c) The procedure of selecting neighbors from the candidates and connecting the new vertex into existing graphs, taking Layer 1 for instance. }
    \label{fig:insertion}
    \Description{WoW insertion}
\end{figure*}

The hierarchical window graphs are constructed layer by layer, from an empty graph at the beginning.

In Algorithm~\ref{alg:insert}, when the top-layer window fails to cover the inserted attribute set (Line 2), the top layer would be raised by cloning the entire old-top to the newly allocated layer (Lines 3--4).
Next, starting from the \textit{top} layer, we compute the window of $v_a$ in Layer $l$, $W_a^l$, which covers $2o^l$ vertices whose attribute values are halved by $a$ (Line 6).
This can be determined by WBT with logarithmic time complexity (see Appendix~\ref{appendix:A}).
In Lines 8--10, the nearest candidates can be found from two sources:
\begin{enumerate*}
    \item In-window candidates of the previous layer, $U^{l+1}$, if the number of them is more than $m$.
    \item Newly retrieved candidates by the beam search procedure on the existing graph, merged with inadequate previous-layer in-window candidates.
\end{enumerate*}
We prove in Theorem~\ref{theo:window_nn} that in higher layers of WoW, candidates are averagely closer to $v_a$ than those in lower layers.
So, if some candidates in the previously processed layers are reserved after filtering for the current layer in Line~8, they should exist in the result given by $\mathrm{SearchCandidates}$ (Algorithm~\ref{alg:search_candidates}) for the current layer.
Thereby, $\mathrm{SearchCandidates}$ can be skipped for faster indexing.

In Line 11, as suggested in \cite{hnsw}, $\frac{m}{2}$ nearest neighbors without dominance relation are selected from the candidates.
The rest $\frac{m}{2}$ empty slots are reserved for latecomers.
For each neighbor $v_b$, we append $v_a$ into the neighbor list if the list is not full (Line 13).
Otherwise, a two-stage pruning procedure is triggered.
It first recalculates the window of $v_b$ and discards neighbors outside the window (Line 16).
Then for the rest of them, RNGPrune works as the second-stage pruning to get the final neighbor list of $v_b$ (Line 17).
Notice that some neighbors may fall out of window after insertion of $v_a$ in Lines 13--14.
They are not pruned immediately for two reasons:
\begin{enumerate*}
    \item Although they are out of the window, they may be in the query range if later insertions do not prune them.
    \item After insertion of more vertices, they may get back into the window and become valid edges again.
\end{enumerate*}
At last, the attribute value $a$ is inserted into WBT and relevant vertices are connected with each other (Line 18).
\begin{theorem} \label{theo:window_nn}
    Without pruning dominated vertices in Line 11 and Line 17, for vertex $v_a$, we have $\sum_{s\in N_{v_a}^{l+1}}\delta(s,v_a)\le \sum_{t\in N_{v_a}^l}\delta(t, v_a)$.
\end{theorem}
\begin{proof}
    Without RNGPrune, $N_{v_a}^l$ and $N_{v_a}^{l+1}$ contain the nearest vertices to $v_a$ covered by $W_a^l$ and $W_a^{l+1}$, respectively.
    As window $W_a^l$ is a subset of window $W_a^{l+1}$, there exists a subset $W' \subseteq W_a^{l+1}\backslash W_a^l$ s.t. $\forall j \in W', \delta(v_j,v_a) < \min_{i\in W_a^l} \delta(v_i,v_a)$.
    
    In addition, $\mathrm{SearchCandidates}$ traverses on the graph in a greedy manner, thus for layer $l+1$ it first expands the candidate list from $W'$ and later from $W_a^l$ before the list is full.
    The inequality holds.
    $\sum_{u\in N_{v_a}^l}\delta(v_a,u) = \sum_{u\in N_{v_a}^{l+1}}\delta(v_a, u)$, if and only if $W' = \emptyset$.
    
    The conclusion also holds with a high probability when RNGPrune is not omitted, as the pruning algorithm runs in a greedy manner and the nearest neighbors are often reserved \cite{graph-survey-2,hnsw,nhq-nips23}.
\end{proof}

Figure~\ref{fig:insertion} depicts an example of inserting $v_{73}$ into the index built in Figure~\ref{fig:struct}.
First, we calculate the windows of $v_{73}$ in all layers using WBT (Figure \ref{fig:insertion}a), where $W_{73}^2=[10,99],W_{73}^1=[48,99],W_{73}^0=[72,74]$.
With these windows, as shown in Figure~\ref{fig:insertion}b, we can acquire candidates for $v_{73}$ in all window graphs using Algorithm~\ref{alg:search_candidates} (which will be presented shortly).
In Figure \ref{fig:insertion}c, once we have acquired the candidates, $\frac{m}{2}=2$ neighbors of $v_{73}$ are picked for each layer, and the neighbors of neighbors are adjusted if necessary.
In Layer 1, $v_{74}$ and $v_{98}$ are chosen, and $v_{60}$ and $v_{99}$ are discarded for adequate out-edges.
$v_{99}$ is also dominated by $v_{74}$ as $v_{73}\rightarrow v_{99}$ is the longest edge in triangle $\Delta v_{73}v_{74}v_{99}$.
In the perspective of $v_{98}$, $v_{60}$ is outside the window after the insertion, but pruning it can be postponed because there is empty space in the neighbor list.

\subsection{Multi-layer Candidate Acquisition}

As shown in Algorithm~\ref{alg:search_candidates}, WoW employs beam search to approach the neighborhood of the target vector following many graph-based indices \cite{hnsw,nsg,diskann,problistic-routing,routing-guided-pq}.
It starts with an entry \textit{ep}, whose corresponding attribute value satisfies the range filter $R$.
Beam search only traverses in layers within a range $L$ with a fixed beam width $\omega$ (a.k.a. $ef$ in HNSW \cite{hnsw}).
Lastly, the nearest neighbors to target $v$ with attribute values all in range $R$ are returned.

\begin{algorithm}[t]
\SetKw{KwOr}{or}
\SetKw{KwAnd}{and}
\SetKw{KwTrue}{true}
\SetKw{KwFalse}{false}
\SetKw{KwNot}{not}
\SetKw{KwBreak}{break}
\SetKwData{KwHyper}{Hyperparameter}
\SetKwInOut{HyperIn}{Hyperparameter}
\caption{SearchCandidates}
\label{alg:search_candidates}
\KwIn{\textit{ep}: graph entry; $v$: target vector; $R$: range filter; $L=[l_{\min},l_{\max}]$: layer range; $\omega$: beam search width}
\KwOut{$U$: nearest candidates}
\HyperIn{$m$: defined in Table~\ref{tab:notation}}
$C\leftarrow \textit{ep}$\Comment*{candidate min-heap during beam search}
$U\leftarrow \textit{ep}$\Comment*{result max-heap of nearest neighbors}
\While{$|C|\neq 0$}{
    $l\leftarrow l_{\max}$, $c_n \leftarrow 0$, $\textit{next}\leftarrow\KwTrue$\Comment*{next: early-stop flag}
    $s\leftarrow$ the nearest vertex to $v$ in $C$\Comment*{current hop}
    \lIf{$\delta(s, v) > \max_{u\in U} \delta(u, v)$}{\KwBreak}
    
    \While(\Comment*[f]{top-down traversal}){$l\ge l_{\min}$ \KwAnd $\textit{next}$}{
        $\textit{next}\leftarrow\KwFalse$\Comment*{shall we check the next layer?}
        \ForEach{unvisited neighbor $v_j\in N_s^l$}{
            \lIf{$j \not\in R$}{$\textit{next}\leftarrow\KwTrue$}
            \ElseIf(\Comment*[f]{discover candidates}){$c_n \le m$}{
                Mark $v_j$ as visited, $c_n \leftarrow c_n + 1$\;
                $t\leftarrow$ the farthest vertex to $v$ in $U$\;
                \If{$|U| < \omega$ \KwOr $\delta(v_j, v) < \delta(t, v)$}{
                    Add $v_j$ into $C,U$\; 
                    \lIf{$|U| > \omega$}{pop from $U$}
                }
            }
        }
        $l \leftarrow l - 1$\;}}
\Return $U$\;
\end{algorithm}

At each hop, the nearest candidate $s$ is selected from $C$ (Line~5).
If the potential neighbors of $v$ are not fully examined (Line 6), the neighbors of $s$ should be checked in a top-down manner (Line~7).
As suggested in Theorem~\ref{theo:window_nn}, we give priority to high-layer unvisited neighbors (Line 9), and push them into $C,U$ (Line 15) if their distance is smaller than the greatest in $U$ (Line 14).
In Line 4, $c_n$ is used to record how many distance computations are made, similar to \cite{acorn,irange}.
It informs WoW to jump to the next hop if $m$ neighbors are checked (Line 11).
The flag \textit{next} is used as an efficient early-stop strategy to avoid unnecessary low-layer visit.
We notice that if all neighbors are in range, despite whether they have been visited, we do not need to check layers below because neighbors are adequate at the current hop (Line 10).
If the early-stop flag fails to prevent low-layer checks, which means that some neighbors are filtered and the search is not sufficient, we go deeper to find more candidates (Line~17).
The search terminates when $C$ is exhausted (Line 3).

Figure~\ref{fig:insertion}b gives a candidate acquisition example with layer range $L=[1,2]$, range filter $R=[48,99]$ (i.e. $W_{73}^{1}$), target vector $v_{73}$, and beam search width $\omega=4$.
The random entry $ep$ is set to $v_{72}$.
The neighbors of $v_{72}$ in Layer 2 are $v_{10}, v_{35}$.
As they are filtered out by $R$, the algorithm then checks Layer 1 and appends the only in-range neighbor $v_{74}$ into the candidate list.
At the 2-hop $v_{74}$, it finds an in-range neighbor $v_{60}$ in Layer 2 and $v_{99}$ in Layer 1. 
Since $v_{60}$ is closer to $v_a$, it is picked as the 3-hop whose neighbor $v_{98}$ in Layer 1 is added into $C,U$ at this hop.
The last hop is $v_{99}$ and there are no more candidates in $C$ because other vertices are far away from $v_a$.
The search converges and the algorithm returns $v_{60},v_{74},v_{98},v_{99}$ as the candidates, with attribute values all in range $[48,99]$.

\subsection{Selectivity-aware Range Search} \label{sec:method-searchknn}

\begin{algorithm}[!t]
\SetKw{KwOr}{or}
\SetKw{KwTrue}{true}
\SetKw{KwFalse}{false}
\SetKw{KwNot}{not}
\SetKw{KwBreak}{break}
\SetKwData{KwHyper}{Hyperparameter}
\SetKwInOut{HyperIn}{Hyperparameter}
\SetNoFillComment
\caption{SearchKNN}
\label{alg:knn}
\KwIn{$(q,R)$: query vector with range; $k$: number of nearest vectors; $\omega_s$: beam search width for query}
\KwOut{$k$ nearest neighbors}
\HyperIn{$o$: defined in Table~\ref{tab:notation}}
\Comment{Step 1: decide landing layer}
$n' \leftarrow$ number of vertices in $R$\Comment*{calculated by WBT}
$l_h \leftarrow \lfloor\log_o\frac{n'}{2}\rfloor$\Comment*{$|W^{l_h}|\le n' < |W^{l_h+1}|$}
$l_d \leftarrow \arg\max_{l\in\{l_h, l_h + 1\}}\frac{\min(2o^l, n')}{\max(2o^l, n')}$\Comment*{landing layer $l_d$}
$\textit{ep}\leftarrow$ vertex with attribute value closest to the median of $R$\;
\Comment{Step 2: acquire multi-layer candidates}
$U\leftarrow\text{SearchCandidates}(\textit{ep}, q, R, [0,l_d], \omega_s)$\;
\Return $k$ nearest vectors to $q$ in $U$\;
\end{algorithm}

To answer an RFANNS query $(q,R)$, Algorithm~\ref{alg:knn} employs two steps.
\begin{enumerate*}
\item The above candidate acquisition procedure requires to decide a feasible layer range $L=[l_{\min},l_{\max}]$.
For $l_{\min}$, we can simply set it to Layer 0 and let the early-stop strategy in Algorithm~\ref{alg:search_candidates} to decide the actual lowest layer at each hop.
For $l_{\max}$, it is acceptable to set it to \textit{top}, as Theorem~\ref{theo:window_nn} suggests that neighbor lists in high layers may have closer candidates than those in lower layers.
However, windows in these layers are large and most neighbors are out of range, which may increase the overhead of many unnecessary filter checks.
Thus, we want to choose a \emph{landing layer} $l_d$ that contains most in-range proximate vertices to reduce filter checks and distance computations simultaneously.
Such a layer resides in either the one with the largest window whose size is less than $n'$ (the number of in-range vertices), namely $l_h$ (Line 2), or $l_h + 1$ (Line~3).
So, $n'$ becomes essential, which also reflects the \emph{selectivity} of range filter in Definition~\ref{def:selectivity}.
Here, WBT functions as an order statistic tree \cite{order-statistic-tree} to give an accurate $n'$ in logarithmic time (see Appendix~\ref{appendix:B}).
\item We retrieve the most nearest $\omega_s$ in-range vectors with Algorithm~\ref{alg:search_candidates} and return the top-$k$ final results.
\end{enumerate*}

\subsection{Influence of Hyperparameters} \label{sec:method-param-impact}

Like other graph-based indices \cite{graph-survey-2,hnsw,nsg}, WoW has two important hyperparameters in construction: maximum outdegree $m$ and construction beam search width $\omega_c$.
$m$ determines the density of the graph and distance computations at one hop.
The optimal $m$ varies with the vector distribution according to~\cite{graph-survey,hnsw,diskann,nhq-nips23,ung}.
Higher $\omega_c$ yields more candidates for neighbor selection, leading to more diversified connections at the cost of longer indexing time \cite{graph-survey, hnsw}. 

The window boosting base $o$ controls the number of layers.
This, in turn, affects indexing speed.
Furthermore, Theorem~\ref{theo:o_impact} proves that $o$ can also impact the proportion of in-range neighbors of a certain vertex on the landing-layer search path.
Recall that it is the optimal layer in terms of the balance between filter checks and distance computations as described in Section~\ref{sec:method-searchknn}.
More in-range neighbors in the landing layer would benefit query with faster navigation speed \cite{nsg,nssg, hnsw}.
Therefore, Theorem~\ref{theo:o_impact} motivates us to select a suitable $o$ for fast indexing and optimal query performance.

\begin{theorem} \label{theo:o_impact}
    In the landing layer $l_d$, the expected fraction $f_R$ of in-range neighbors at a single hop on the path is bounded by
    \begin{equation}
        f_R=\begin{cases}
            (\frac{1}{\sqrt{o}},\frac{1}{2}) & \text{(a) } l\in (l'-1,l'-\frac{1}{2}), o>4, \\
            \left[\frac{\sqrt{2}}{2}-\frac{1}{4o^{l+1}},\frac{3}{4}-\frac{1}{4o^{l+1}}\right) &  \text{(b) } l\in (l'-1,l'-\frac{1}{2}), o \leq 4, \\
            \left[\frac{3}{4}-\frac{1}{4o^{l}}, 1-\frac{o^l + 1}{4o^{l + \frac{1}{2}}}\right] &  \text{(c) } l \in [l'-\frac{1}{2},l'],
        \end{cases}
    \end{equation}
    where $n'$ is the number of in-range vectors for filter $R=[x,y]$, $l'=\log_o\frac{n'}{2}$, and $l = l_h = \lfloor\log_o\frac{n'}{2}\rfloor$, which is the highest layer defined in Line~2 of Algorithm~\ref{alg:knn}, whose window size is less than $R$.
\end{theorem}
The proof is provided in Appendix~\ref{appendix:C}.
For example, when $o=2, n'=\text{2,048}$, we have $l = l' = 10$, which satisfies Case (c).
In this setting, the lower bound is 74.97\% and the upper is 82.30\%, leading to 12--13 in-range neighbors at a single hop when $m=16$.
If $o$ is a large value, the construction would be fast, but the lower-bound in Case (a) would be small.
To achieve the optimal fraction lower bound in all cases, while maintaining good indexing performance, setting $o=4$ is recommended, as verified in Section~\ref{sec:exp-detailed}.

\subsection{Complexity Analysis}\label{sec:complexity}
\paragraph{Index size.} The index size is twofold. 
\begin{enumerate*}
\item WBT costs $O(cn)$ space, where $c$ is a constant overhead of the tree storage.
\item Neighbor lists in the hierarchy of window graphs occupy most of the space. 
A vertex in Layer $l$ has an average outdegree of $\min(2o^l, m)$. As the top layer is $\textit{top} = \lceil \log_o\frac{n}{2} \rceil$, the total space for all link lists is $n\sum_{l=0}^{\textit{top}} \min(2o^l,m)$.
\end{enumerate*}
Considering the two parts, \emph{the overall asymptotic space complexity is $O(cn+mn\lceil \log_o\frac{n}{2} + 1\rceil)$.}

\paragraph{Insertion time.}
In Lines 1--4 of Algorithm~\ref{alg:insert}, the amortized complexity to raise the top layer is $O(\lceil\log_o\frac{n}{2}\rceil)$.
In Lines 5--18, the vector would be inserted into $\lceil\log_o\frac{n}{2} + 1\rceil$ layers.
In Layer $l$, it costs $O(\log_{\alpha'}n)$ to locate the window boundaries (by searching in WST, a.k.a. BB[$\alpha$] tree), where $\alpha'=\frac{1}{1-\alpha}$.
Then in Line 10, candidates are retrieved from the window graph inside $W_a^l$ with size $2o^l$. In high-dimensional space, the complexity of beam search scales to $O(\log 2o^l)$ on RNG-based indices \cite{hnsw,nsg,nssg,graph-survey,tao-mg,graph-survey-2}.
The \emph{worst-case} time complexity of candidate acquisition occurs in the situation when the $\mathrm{SearchCandidates}$ procedure is invoked for all layers, and the result is $\sum_{l=0}^{\textit{top}}O(\log 2o^l)= O(\lceil\log_on\rceil^2\log 2o)$.
RNG neighbor selection for $v$ and two-stage pruning for the neighbors of $v$ in Lines 11--17 have sublinear time complexity, $O(\frac{m}{2}\log\omega_c+m^2\log m +m\log_{\alpha'}n)$.
The insertion of WBT scales to $O(\log_{\alpha'}n)$ with amortized constant number of self-balancing operations.
\emph{In summary, the worst-case insertion time complexity scales to $O(\log^2n)$.}

\paragraph{Query time.}
For landing layer selection, the cardinality of the filtering subset by the range filter $[x,y]$ can be calculated by the rank $i,j$ of the upper bound of $x$ and the lower bound of $y$ from WBT, which requires $O(\log n)$ time separately.
Then, $n'$ can be calculated by $j - i + 1$.
The complexity of candidate acquisition is $O(\log n')$.
\emph{The overall query time complexity scales to $O(\log n')$.}

\subsection{Further Extensions} \label{sec:extensions}

\paragraph{Duplicate attribute values.}
Algorithm~\ref{alg:search_candidates} is agnostic to value redundancy and relies solely on the query range.
Duplicate values have the same rank in Definition~\ref{def:window_graph}, and thus only vectors are considered.

Specifically, duplicate values are not inserted into WBT. 
Only their corresponding vectors are inserted into the window graphs.
The condition of Line 2 in Algorithm~\ref{alg:insert} is changed to $|\mathcal{A}|_u + 1 > 2 o^\textit{top}$, where $|\mathcal{A}|_{u}$ denotes the number of unique values.
We define \textit{unique selectivity} as the ratio of distinct values in the filtered set to those in $\mathcal{A}$.
In Step~1 of Algorithm~\ref{alg:knn}, the layer with window size closest to the number of in-range unique values is selected, aligning with the filter's unique selectivity.
With duplicate values, the number of layers reduces from $\lceil \log_o\frac{n}{2} + 1\rceil$ to $\lceil \log_o\frac{|\mathcal{A}|_u}{2} + 1\rceil$.
This leads to proportionally lower space and insertion time complexity.

\paragraph{Deletion and in-place update}
For a single deletion, we can mark the deleted vertex and normally traverse it without pushing it into the result max-heap.
It will be removed from the neighbor list if the two-stage pruning procedure is triggered.
The query complexity stays unchanged if the amount of deletions is not significant.
For frequent deletions and in-place updates, existing methods \cite{topology-repair,fresh-diskann,gti} can be integrated into all graph-based indices, including WoW.

\paragraph{Multi-attribute RFANNS}
It is still an open problem to design a dedicated multi-attribute index.
Existing works discuss several rudimentary solutions \cite{serf,irange,hsig}: 
\begin{enumerate*}
    \item Building a composite index and querying on it via lexicographic order;
    \item Decomposing one query into sub-queries on corresponding single-attribute indices before merging the separate results;
    \item Querying on a single-attribute index for one predicate and applying in-filtering or post-filtering on others.
\end{enumerate*}
There are two future directions to support multi-attribute RFANNS in WoW:
\begin{enumerate*}
    \item For the second existing solution, sub-queries can be conducted in one round by composing relevant layers from multiple single-attribute indices.
    Our 1D window can also be extended to high-dimensional rectangles for disjunctive range queries \cite{geo-rfanns};
    \item If the filter has complex predicates on several attributes, a learned cardinality estimator~\cite{dynamic-cardest,cardest-survey} for these predicates can be employed to assist WBT to guide the layer selection.
\end{enumerate*}

\begin{table*}[t!]
    \centering
    \caption{Statistical data of the used datasets}
    \label{tab:data}
    {\small
    \begin{tabular}{lcccccr}
        \toprule
        Dataset & Modality & Size & Dimension & Distance metric & Attribute type &  LID@10 $(2^{0}\,/\,2^{-1}\,/\,2^{-4}\,/\,2^{-7}\,/\,2^{-10})$\\
        \midrule
        Sift & Image & \ \ 1,000,000 & \ \ \,128 & Euclidean & Random integer & 28.3 / 31.7 / 28.8 / 26.0 / 22.3\\
        Gist & Image & \ \ 1,000,000 & \ \ \,960 & Euclidean &  Random integer & $\bigstar$ 66.6 / 64.9 / 57.3 / 48.5 / 37.6\\
        ArXiv & Text & \ \ 2,138,591 & \ \ \,384 & Cosine & Time & 12.6 / 13.0 / 15.5 / 17.7 / 17.9\\
        Wikidata4M & Text & \ \ 4,000,000 & 1,024 & Cosine & Title & $\bigstar$ 25.3 / 26.7 / 27.9 / 29.4 / 29.6 \\
        Deep10M & Image & 10,000,000 & \ \ \ \ \,96 & Euclidean & Vector ID &  32.8 / 31.6 / 30.1 / 27.5 / 24.2\\
        \bottomrule
    \end{tabular}}
\end{table*}

\section{Experiments and Results} \label{sec:exp}
In this section, we evaluate WoW in index construction and query by answering the research questions below:
\begin{itemize}
\item[\bf RQ1.] Is the parallel incremental construction of WoW efficient? 
Are its indexing time and index size competitive?

\item[\bf RQ2.] Does WoW surpass existing methods in RFANNS querying? 
How is the performance under different query selectivity?

\item[\bf RQ3.] Can our optimizations, e.g., RNG pruning, early-stop and landing layer selection strategies, improve performance?
How are the index extensibility and parameter sensitivity?
\end{itemize}

\subsection{Experiment Setting} \label{sec:exp-setting}

\paragraph{Environment.}
All experiments are conducted on a Ubuntu 20.04 server with an Intel Xeon E5-1607 v3 (3.10GHz) CPU and 240GB RAM.
The CPU has 128K L1 cache, 10M L3 cache, and SIMD support.
All methods are implemented in C++ and compiled by GCC 9.4.0.

\paragraph{Datasets.}
We pick five real-world datasets with different characteristics.
Table~\ref{tab:data} lists their statistical data.
\begin{itemize}
\item \textbf{Sift} and \textbf{Gist}\footnote{\url{http://corpus-texmex.irisa.fr}} are two datasets embedded from images to assess the query performance of ANN indices.
Following the practice of prior works~\cite{acorn,serf,irange,wst,nhq-nips23}, we generate a random integer for each vector as its attribute value.

\item \textbf{ArXiv}\footnote{\url{https://github.com/qdrant/ann-filtering-benchmark-datasets}} is a recent hybrid dataset released by QDrant. 
It involves vectors embedded from over two million paper titles using the all-MiniLM-L6 encoder, along with release time as the attribute.

\item \textbf{Wikidata4M}\footnote{\url{https://huggingface.co/datasets/Cohere/wikipedia-2023-11-embed-multilingual-v3}} is embedded from Wikipedia entity descriptions by the Cohere Embed v3 model.
We use the first four million vectors, with entity title as the attribute.

\item \textbf{Deep10M}\footnote{\url{https://research.yandex.com/blog/benchmarks-for-billion-scale-similarity-search}} is obtained from the last fully-connected layer of the GoogleNet model \cite{deep}.
We use the first 10 million vectors, with vector ID as the attribute.
\end{itemize}

For simplicity, we follow previous works \cite{serf,wst,irange} to first sort the vectors by attribute values and then use vector IDs as the new attribute values.
The new vector-attribute pairs are reshuffled to simulate the \emph{unordered insertion} scenario.

Each query vector is assigned a query range, which is randomly generated by the fraction of in-range vectors over the dataset size $f$ (cf. \textsc{Definition}~\ref{def:selectivity}).
For example, the query range fraction $f=2^{-6}$ on a dataset with 1,000 vectors indicates that there exist $\lfloor1000\times 2^{-6}\rfloor = 15$ in-range vectors, and a corresponding query range can be $[1,15]$, $[3,17]$, etc.
We refer to ranges with fractions in $[2^{-10}, 2^{-9}]$ as exhibiting extreme selectivity, $[2^{-8},2^{-6}]$ as high selectivity, $[2^{-5},2^{-3}]$ as moderate selectivity, and $[2^{-2},2^{0}]$ as 
low selectivity.
We further assemble a \emph{mixed} range workload comprising query ranges with equal number of fractions in $[2^{-10}, 2^{0}]$.

In our experiments, each query workload contains 1,000 range-filtering queries.
We estimate the hardness of a query workload by Local Intrinsic Dimensionality (LID) \cite{ann-survey,estimate-lid,steiner-hardness,role-lid}, as defined below:
\begin{definition}[Local Intrinsic Dimensionality, LID] \label{def:lid}
    Given a hybrid dataset $\mathcal{D} =\{\mathcal{V}, \mathcal{A}\}$ and a range-filtering query workload $Q = \{\mathcal{V}_Q, \mathcal{R}\}$, the LID of $Q$ is estimated by each query $(v,R)$'s top-$k$ in-range nearest neighbors $(u_i,a_i)\in\left\{u_i\in N(v) \wedge a_i \in R\right\}$:
    \begin{equation}
        LID@k=\mathbb{E}_{(v,R) \in Q} \left[ -\left( \frac{1}{k}\sum_{i=1}^{k} \log\left(\frac{\delta(v, u_i)}{\delta(v,u_k)}\right) \right)^{-1} \right].
    \end{equation}
\end{definition}

In Table~\ref{tab:data}, we show LID@10 of the fractions $2^{0},2^{-1},2^{-4},2^{-7}$, and $2^{-10}$. 
A larger LID indicates a more challenging dataset. 
For the datasets defined under the same distance metric (e.g., Sift, Gist, and Deep10M), the hardest one is marked with ``$\bigstar$'' (e.g., Gist).

\paragraph{Competitors and parameters.}
We choose six RFANNS indices, WST \cite{wst}, iRangeGraph \cite{irange}, SeRF\cite{serf}, HSIG \cite{hsig}, RangePQ \cite{rangepq}, and DIGRA \cite{digra}, two generic predicate-filtering indices, Milvus-HNSW \cite{milvus-vdb} and ACORN \cite{acorn}, as well as two baselines, pre-filtering and post-filtering. 
All indices except pre-filtering and RangePQ are RNG-based, so they share two important hyperparameters in index building: the construction beam search width $\omega_c$ and the maximum outdegree $m$ (for HNSW, $m$ is the bottom-layer maximum degree).
Considering dataset size and hardness (LID@10), we prepare a set of default values for fair comparison: $\omega_c=128,m=16$ for Sift, and $\omega_c=256,m=16$ for the rest.
Following \cite{serf,wst,irange,hsig,acorn,filtered-diskann,nhq-nips23}, all indices find $k=10$ nearest neighbors for each of 1,000 queries.
\begin{itemize}
    \item \textbf{WoW} is our index.
    We set the window boosting base to $o=4$ as analyzed in Section~\ref{sec:method-param-impact}. 

    \item \textbf{SeRF} employs \textit{2DSegmentGraph} to handle arbitrary ranges.
    We use $\textit{leap}=\text{MAX\_POS}$ as suggested in the paper.
    Since SeRF cannot gain satisfactory recall (0.85) under the default hyperparameters, we increase to $\omega_c=512, m=64$.
    
    \item \textbf{WST}. We pick the most performant indexing strategy in it, \textit{SuperPostFiltering}.
    The branching factor is set to $\beta=2$.
    
    \item \textbf{iRangeGraph} does not have additional parameters. 

    \item \textbf{HSIG} splits a dataset into slots according to the attribute.
    We set $\textit{slot\_num}=8$ as studied in its paper.

    \item \textbf{RangePQ}.
    We set $L_{base}$ to 1,000 for Sift, 3,000 for Gist, and 10,000 for the others, as studied in its paper.

    \item \textbf{DIGRA} has no extra parameters.
    Due to query performance degradation (recall < 0.85) caused by insertions, RangePQ and DIGRA are evaluated only under static construction.
    
    \item \textbf{Milvus and ACORN}.
    We use Milvus v2.5 and build HNSW indices, where $\omega_c, m$ are set to default.
    For ACORN, we set $m=64, m_\beta=2m$ to get the best performance, and $\gamma=2^{10}$ to support the highest selectivity as it suggests.
    
    \item \textbf{Pre-filtering and Post-filtering.} Pre-filtering is used to generate the ground truth.
    For post-filtering, we set the size of intermediate result to $s\times k$ ($s$ is defined in Definition~\ref{def:selectivity}).
\end{itemize}

\begin{table}
    \centering
    \caption{Index size and indexing time}
    \label{tab:index_time_size}
    {\small
    \begin{tabular}{lrrrrr}
        \toprule
        Size (MB)      & Sift & Gist & ArXiv & Wikidata4M & Deep10M \\
        \midrule
        HNSW-L0        & 76    & 72    & 163   & 305   & 762 \\
        Milvus & 137 & 137 & 294 & 550 & 1,375\\
        ACORN & 4,832 & 4,832 & 9,821 & 18,371 & 45,967 \\
        \midrule
        SeRF           & \textbf{430} & \textbf{457} & \textbf{1,133} & \textbf{2,385} & 5,486 \\
        WST            & 1,002 & 806   & 3,043 & 5,763 & 16,353 \\
        iRangeGraph    & 869   & 793   & 2,247 & 4,715 & 11,598 \\
        HSIG           & 1,101 & 1,101 & 2,355 & 4,405 & 11,014 \\
        RangePQ & 791 & 3,857 & 3,495 & 16,259 & \textbf{4,293} \\
        DIGRA & 1,533 & 1,533 & 3,547 & 9,501 & 21,477  \\
        WoW     & 713   & 713   & 1,664 & 3,112 & 8,430 \\
        \bottomrule
        \toprule
        Time (s)       & Sift & Gist & ArXiv & Wikidata4M & Deep10M \\
        \midrule
        HNSW-L0        & 17  & 113   & 111   & 351    & 319 \\
        Milvus-HNSW & 25 & 176 & 252 & 918 & 529 \\
        ACORN & 313 & 2,343 & 2,201 & 8,542 & 6,089 \\
        \midrule
        SeRF           & \textbf{112} & \underline{591} & \textbf{592} & \underline{3,053} & \textbf{1,574} \\
        WST            & 257 & 1,688 & 6,418 & 35,394 & 11,602 \\
        iRangeGraph    & 463 & 2,401 & 3,211 & 8,770  & 9,729 \\
        HSIG           & 960 & 2,305 & 4,383 & 20,507 & 12,976 \\
        RangePQ & 566 & 4,508 & 4,105 & 20,494 & 4,493 \\
        DIGRA & 1,935 & 11,236 & 15,230 & 49,178 & 51,312  \\
        WoW (1 thd.) & 765 & 2,713 & 3,728 & 10,371 & 19,013 \\
        WoW (16 thd.)    & \underline{152} & \textbf{391} & \underline{623} & \textbf{1,557} & \underline{3,360} \\
        \bottomrule
    \end{tabular}}
\end{table}

\subsection{RQ1: Index Construction Efficiency} \label{sec:exp-building}

We build DIGRA and RangePQ in one thread as they do not implement parallelism, and other indices with 16 threads, including a single-layer HNSW built by hnswlib (HNSW-L0).
We also build WoW in a single thread, denoted by WoW (1 thd.), for comparison.
Table~\ref{tab:index_time_size} presents the index size and indexing time of the evaluation datasets.
The size of raw vectors and attribute values are excluded.

\begin{enumerate}[wide]
\item For the indices designed for generic predicate, the index size and indexing time of Milvus-HNSW are better than those of ACORN, since Milvus-HNSW is a vanilla HNSW over all vectors, which is fast to build and the index size is on par with HNSW-L0.
ACORN builds a dense HNSW without RNG pruning to ensure query accuracy, leading to the largest index size among all indices.

\item SeRF outperforms all competitors in terms of both index size and indexing time on the easy datasets (i.e. Sift, ArXiv, and Deep10M).
Despite its high indexing efficiency relying on ordered insertion, SeRF fails to achieve stable recall on most query workloads, which is shown shortly in Section~\ref{sec:exp-query}.

\item WoW achieves the second best in indexing time on the easy datasets, and outperforms all indices on the hard datasets, Gist and Wikidata4M.
The acceleration over graph-based indices is from
\begin{enumerate*}[label=(\roman*)]
    \item Its \emph{worst-case} time complexity is on par with the \emph{average} time complexity of others \cite{serf,wst,irange,hsig,digra};
    
    \item Like SeRF, WoW acquires candidates on the incomplete graph.
    The candidate acquisition speed relies heavily on the graph quality.
    SeRF compresses oracle HNSWs and some proximity relations may be lost.
    Differently, WoW has almost identical quality to the oracle graph, as evaluated in Section~\ref{sec:exp-detailed}.
    On the hard datasets, dense neighborhoods around targets make higher-quality graphs beneficial for faster indexing.
\end{enumerate*}
The size and indexing speed of RangePQ are sensitive to dimensions.
It is faster than WoW (1 thd.) on low-dimensional datasets (e.g., Sift and Deep10M) but less competitive on high-dimensional datasets.

\item For a given dataset, the size of WoW is roughly the product of the layer number and the size of HNSW-L0.

If the layer number is consistent, the index size also grows linearly with the number of vectors.
Comparing the indices built on ArXiv and Wikidata4M, their layers are the same, which can be calculated by $\lceil\log_o\frac{n}{2}+1\rceil$. 
The index size on Wikidata4M is about 2$\times$ of that on ArXiv.

\item WoW can achieve 5.0--6.9$\times$ multi-thread acceleration by fine-grained locks at the vertex level \cite{hnsw}.
The parallelism ability enables WoW to scale to datasets with larger size and higher dimensions.
\end{enumerate}

\begin{figure*}
    \centering
    \includegraphics[width=\linewidth]{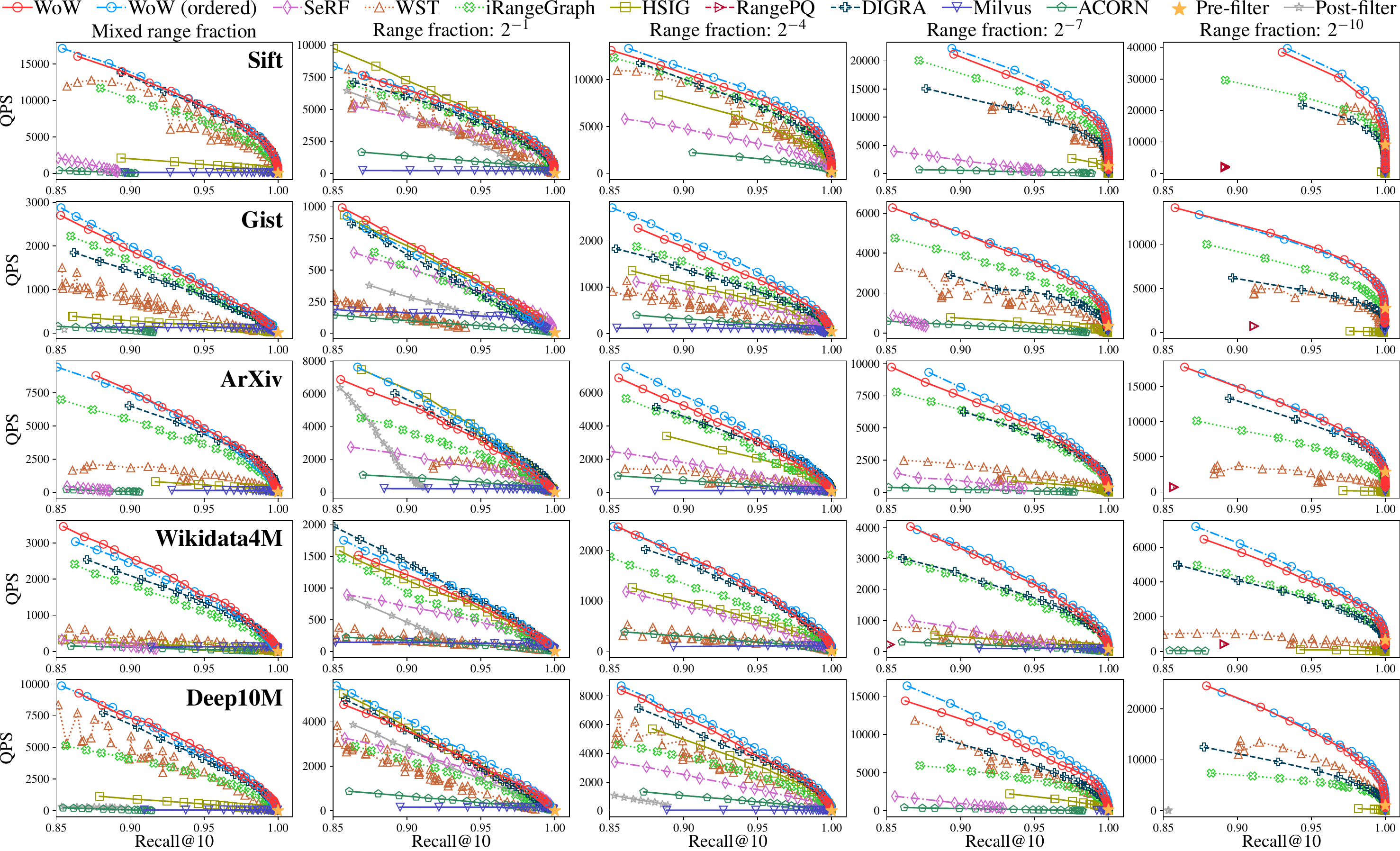}
    \caption{QPS-Recall@10 curves. 
    For each dataset, five workloads with varying range fractions are shown, including the mixed range fraction workloads.
    Dotted line ($\cdot\cdot\cdot$) stands for static index, dashed line ($---$) for post-incremental, dash-dotted line ($-\cdot-$) for ordered incremental, and solid line (---) for unordered incremental. 
    Curves below 85\% recall are omitted.}
    \Description{QPS-Recall@10 curves.}
    \label{fig:qpsrecall}
\end{figure*}

\subsection{RQ2: RFANNS Query Performance} \label{sec:exp-query}

This evaluation is carried out using one single thread.
As shown in Figure~\ref{fig:qpsrecall}, we assess query performance using the Queries Per Second (QPS)-Recall@10 curves \cite{ann-survey,hnsw,ann-benchmarks}.
According to Table~\ref{tab:rfann-comp}, we categorize the competing indices into static, post-incremental, ordered incremental, and unordered incremental.
We also verify WoW with ordered insertion, named WoW (ordered).
Since post-filtering, ACORN, and Milvus-HNSW are not aware of the attribute order, they are regarded as unordered incremental in this experiment.
\begin{enumerate}[wide]
\item WoW surpasses other indices under the mixed range fraction workloads of all datasets.
Compared to the best static indices, WoW obtains about 1.4$\times$ acceleration over iRangeGraph and 6$\times$ over WST at 90\% recall.
Compared to the post-incremental index DIGRA, which is built statically in our experiments, WoW achieves identical or better performance.
It is nontrivial for a dynamic index to maintain competitive performance against statically-built ones, because it has to correctly reorganize the neighbor relations at every new insertion.
Compared to the unordered incremental index HSIG, WoW gains about 4$\times$ query speed-up at 90\% recall.
Compared to the ordered incremental index SeRF, WoW is $6\times$ speedup on the Wikidata4M dataset that SeRF can achieve 90\% recall.

\item WoW (ordered) slightly exceeds WoW under a few workloads (e.g., on Sift).
This is because for ordered insertion, the window of a certain vertex would stay unchanged after the data located at the right window boundary is inserted.
For unordered insertion of a data object, however, windows of vertices that have close attribute values to the inserted one are changing all the time.
Therefore, the first pruning stage is triggered continuously to maintain the window property and some high-quality neighbors may be pruned out.
It may lead to the loss of some neighbors that are supposed to exist in the final window, and a slight decline in query performance compared to WoW (ordered).

\begin{figure*}[t]
    \centering
    \includegraphics[width=\linewidth]{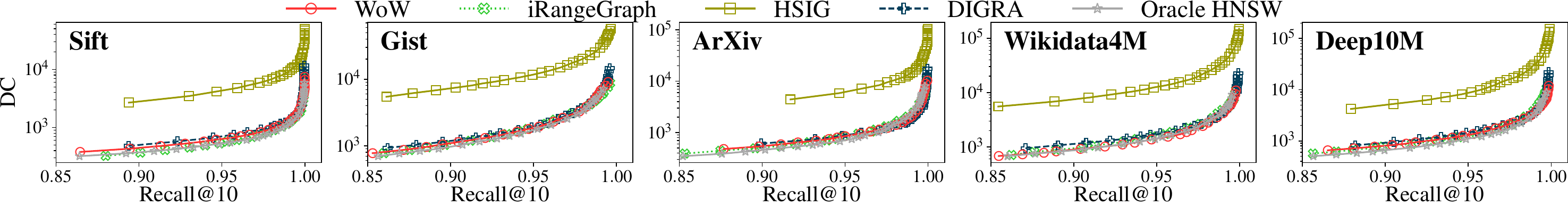}
    \caption{DC-Recall@10 curves compared to oracle HNSW for mixed range fraction.
    Down and to the right is better.}
    \label{fig:oracle}
    \Description{DC-Recall@10 curves compared to oracle HNSW}
\end{figure*}

\begin{table*}[t]
    \centering
    \caption{QPS-Recall@10 and DC-Recall@10 of WoW and WoW without early-stop}
    {\small
    \begin{tabular}{cl*{15}r}
        \toprule
        & \multirow{2}{*}{Recall@10} & \multicolumn{3}{c}{Sift} & \multicolumn{3}{c}{Gist} & \multicolumn{3}{c}{ArXiv} & \multicolumn{3}{c}{Wikidata4M} & \multicolumn{3}{c}{Deep10M}\\
        \cmidrule(lr){3-5} \cmidrule(lr){6-8} \cmidrule(lr){9-11} \cmidrule(lr){12-14} \cmidrule(lr){15-17}
        & & 90\% & 95\% & 99\% & 90\% & 95\% & 99\% & 90\% & 95\% & 99\% & 90\% & 95\% & 99\% & 90\% & 95\% & 99\% \\
        \midrule
        \multirow{2}{*}{QPS} & WoW & 13,993 & 8,740 & 4,563 & 1,815 & 1,043 & 503 & 7,772 & 4,801 & 1,883 & 2,518 & 1,486 & 523 & 6,952 & 4,521 & 1,455 \\
        & w/o early-stop & 10,491 & 7,400 & 3,538 & 1,538 & 842 & 484 & 6,702 & 4,472 & 1,628 & 2,265 & 1,171 & 388 & 6,632 & 4,101 & 1,336 \\
        \midrule
        \multirow{2}{*}{DC} & WoW & 432 & 665 & 1,216 & 1,139 & 1,942 & 3,857 & 534 & 849 & 2,074 & 918 & 1,539 & 4,200 & 883 & 1,350 & 3,919 \\
        & w/o early-stop & 520 & 726 & 1,517 & 1,272 & 2,277 & 3,890 & 617 & 926 & 2,540 & 1,015 & 1,946 & 5,676 & 943 & 1,569 & 4,855 \\
        \bottomrule
    \end{tabular}}
    \label{tab:early_stop}
\end{table*}

\begin{figure}
    \centering
    \includegraphics[width=\linewidth]{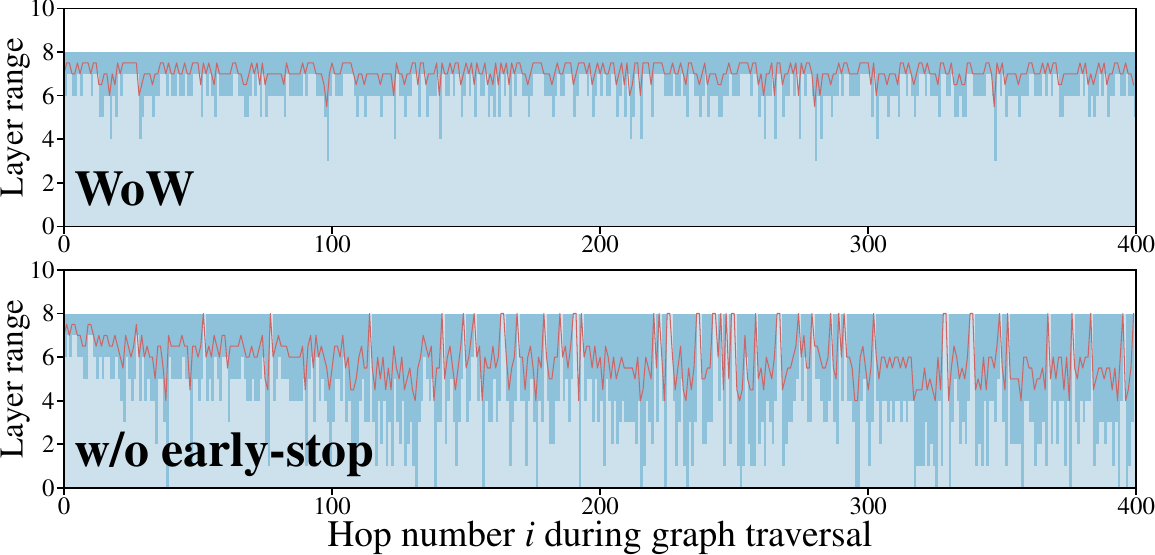}
    \caption{Layer range footprints of WoW and w/o early-stop on a single query of Gist under the range fraction $f = 2^{-4}$.
    The results of the first 400 hops are shown.
    Dark blue bars represent the highest visited layer $l_{max}$ and light ones are the lowest $l_{min}$.
    The gap between two bars is the exploring depth.
    The red curve is the trend of the median as $\frac{1}{2}(l_{\min}+l_{\max})$.}
    \label{fig:hop-layer}
    \Description{Layer range footprints}
\end{figure}

\item Across workloads of varying selectivity, WoW and WoW (ordered) consistently perform the best, while some competitors (e.g., SeRF, HSIG, DIGRA, and iRangeGraph) face performance degradation as the selectivity increases.
As most vectors can pass the range filter under workloads of low selectivity (e.g., $f=2^{-1}$), it is easy for all RFANNS indices to discover neighbors with high recall.
When the selectivity increases, some indices like WST and HSIG need to do more distance computations to achieve high recalls due to less correlation between vectors and attribute values.
SeRF fails to achieve 95\% recall because some in-range neighbors are missing during traversal and falls into local optima, due to the lossy graph compression as reported in \cite{irange,serf}.
iRangeGraph has to make a thorough tree traversal from the root to the low-level nodes at each hop on the search path, leading to diminished query performance under workloads of high selectivity.
DIGRA uses several tree nodes to cover the query range but suffers redundant distance computations if the range is less compatible with tree nodes, causing unstable performance on workloads $f=2^{-7}$ or $2^{-10}$ (e.g., good performance on ArXiv but decreased on others).
WoW benefits from the early-stop strategy in Algorithm~\ref{alg:search_candidates} and the selectivity-aware layer selection in Algorithm~\ref{alg:knn}, to resist redundant filter checks and distance computations for workloads of high selectivity.

\item On the hard datasets Gist and Wikidata4M, WoW, WoW (ordered), and iRangeGraph achieve superior QPS than other indices in terms of the same range fraction.
Vectors in these datasets tend to crowd around the query vector, making the graph-based algorithms harder to determine the true neighbors \cite{ann-survey,graph-survey}.
iRangeGraph can smartly decide from which tree node to choose neighbors at each hop.
WoW has similar or even better graph quality than iRangeGraph, revealed in Section~\ref{sec:exp-detailed}.
In these two indices, the selected edges together compromise a high-quality RNG to withstand the hardness of the datasets.
Due to the lack of navigation edges and less accurate proximity relations, other indices are struggling to achieve high query speed at high recalls (e.g., $>95\%$).
For the easy datasets, all indices can perform well because vectors are distributed sparsely, making it easier to determine the nearest neighbors.

\item Under all workloads, graph-based RFANNS indices generally outperforms those for generic predicates.
Milvus can achieve stable accuracy due to its effective query optimizer to select between pre-filtering and post-filtering.
ACORN is good at queries of moderate selectivity compared to Milvus, as it only calculates distance for in-range vectors, better than the post-filtering strategy of Milvus.
However, its performance drops when the selectivity increases.
This is because out-of-range vertices are the majority, causing beam search to end prematurely before obtaining enough candidates.

\item It is worth noting that SeRF cannot achieve 85\% recall rate under some mixed workloads, due to its lossy compression of multiple HNSWs.
RangePQ has reasonable recall under workloads with high and extreme selectivity (e.g., 88\% and 91\%), but struggles to obtain high recall under low and moderate ones, which is common to PQ-based indices and in accord with the conclusion in its paper.
As post-incremental indices, both RangePQ and DIGRA encounter recall loss after bulk insertions.
For example, after building on 50\% of vectors by static construction and inserting the rest half, the highest recall that DIGRA can achieve on Sift with range fraction $f=2^{-1}$ falls from 99\% to 27\%, due to the broken interconnections of graphs by tree node splitting.
The performance of WST is not stable and the curves fluctuate with recall.
In the post-filtering phase of WST, the query performance is affected by an extra parameter beyond $\omega_s$.
It is used to adjust the intermediate result size, and with a fixed $\omega_s$, a larger size can enhance accuracy.
Thus, WST brings in an additional challenge to parameter tuning.
iRangeGraph encounters a prominent performance drop on the Deep10M dataset.
Due to $O(\log n + m)$ neighbor-selection overhead per hop, it cannot perform well on low-dimensional and large-scale datasets, where the distance computations are fast or the search paths are long, consistent with the observation in its paper \cite{irange}.
\end{enumerate}

\begin{figure}
    \centering
    \includegraphics[width=\linewidth]{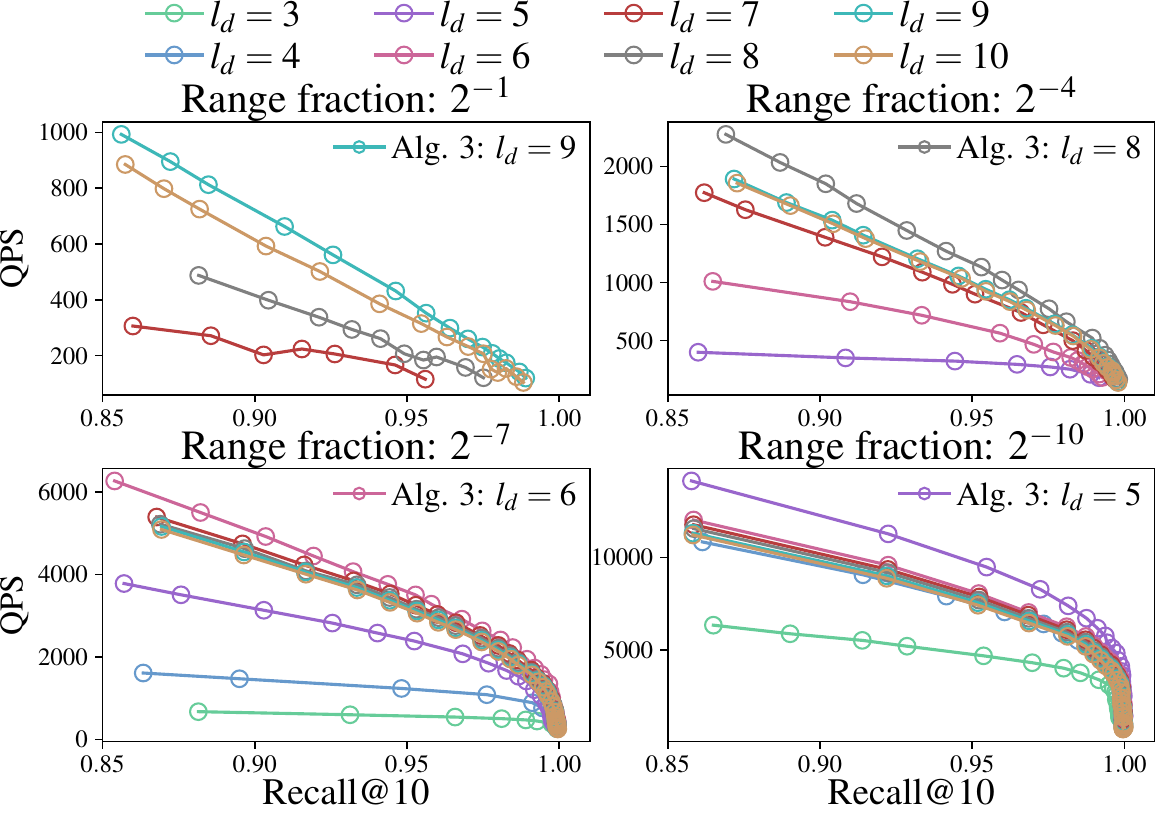}
    \caption{Effectiveness of landing layer selection. 
    The layers selected by Algorithm~\ref{alg:knn} are given on the top right corner.
    Some curves are absent for failing to achieve 85\% recall.}
    \label{fig:layer_sel}
    \Description{Effectiveness of landing layer selection}
\end{figure}

\subsection{RQ3: Detailed Analysis} \label{sec:exp-detailed}

\paragraph{Approximation to oracle RNG}
The performance of graph-based RFANNS indices is bounded by the oracle proximity graph \cite{acorn,irange}.
The approximation quality of an RFANNS index to the oracle graph can be evaluated by comparing the Distance Computations per Query (DC) between the two.
We consider indices using NSW or HNSW as their underlying graph, and construct oracle HNSWs for each query range with identical $\omega_c$ and $m$.
The number of distance computations performed on these oracle HNSWs represents the lower bound achievable by the most accurate RFANNS indices.

As shown in Figure~\ref{fig:oracle}, the DC-Recall@10 curves of WoW coincide closely with those of oracle HNSW, especially on the harder datasets and for recall above 95\%.
iRangeGraph and DIGRA also achieve optimal approximation to oracle HNSW.
However, DIGRA encounters accuracy loss due to insertions, while iRangeGraph fails to achieve close QPS to that of WoW, due to its neighbor selection overhead as reported in Section~\ref{sec:exp-query}.
Moreover, they have full horizon of the entire dataset during static construction to finetune the neighbor proximity, while WoW is only aware of the inserted vectors with limited global information.
HSIG supports incremental construction, but it cannot obtain a comparable approximation.

\begin{figure}[t]
    \centering
    \includegraphics[width=\linewidth]{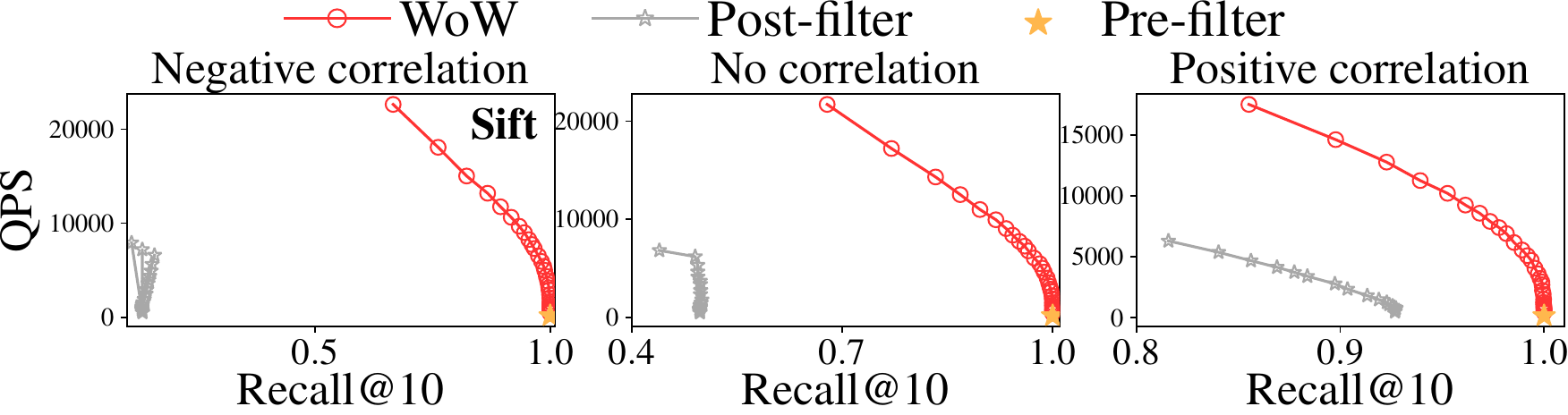}
    \caption{QPS-Recall@10 under workloads with different query correlations on the Sift dataset}
    \label{fig:query-correlation}
    \Description{QPS-Recall@10 under workloads with different query correlations on the Sift dataset}
\end{figure}

\begin{figure}[t]
    \centering
    \includegraphics[width=\linewidth]{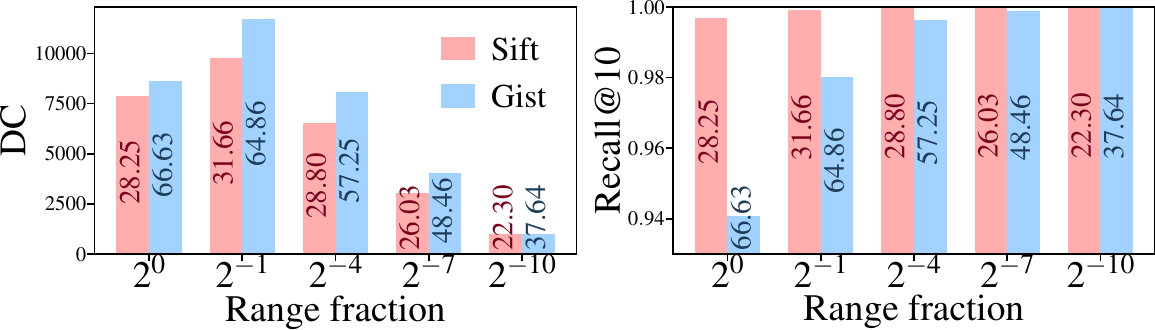}
    \caption{DC and Recall@10 for different LID@10 (marked in the middle of the bars)}
    \label{fig:dc-lid}
    \Description{DC and Recall@10 for different LID@10}
\end{figure}

\paragraph{Early-stop strategy.}
Table~\ref{tab:early_stop} depicts the query performance of WoW and WoW without early-stop in Algorithm~\ref{alg:search_candidates}.
WoW outperforms WoW w/o early-stop in terms of QPS and DC.
This is because without early-stop, WoW has to search more on lower layers.
As demonstrated in Theorem~\ref{theo:window_nn}, vectors in lower layers are less accurate than those in higher ones in terms of vector proximity.
It results in longer search paths and more unnecessary distance computations that may diminish query speed.
The early-stop strategy can also avoid some low-layer filter checks to further improve QPS.

To further demonstrate the effectiveness of the early-stop strategy, Figure~\ref{fig:hop-layer} shows the layer range footprints of WoW for a single query with $f=2^{-4}$ on the hard Gist dataset.
The layer range at a single hop is represented as $[l_{\min},l_{\max}]$. 
According to the landing layer selection strategy in Algorithm~\ref{alg:knn}, $l_{\max}$ is set to 8 instead of the top layer 10.
$l_{\min}$ is the lowest layer that Algorithm~\ref{alg:search_candidates} reaches.
The early-stop strategy can bound the layer exploration within 1--2 layers deeper at each hop, while allowing lower layer checks if current exploration is not adequate.
Without early-stop, $l_{\min}$ is not bounded, therefore 4--5 lower layers are checked for most hops.
In summary, the query speed acceleration is twofold:
\begin{enumerate*}
    \item Avoiding more distance computations for less accurate vectors in low layers.
    \item Avoiding unnecessary filter checks in low layers. 
\end{enumerate*}

\begin{figure}[t]
    \centering
    \includegraphics[width=\linewidth]{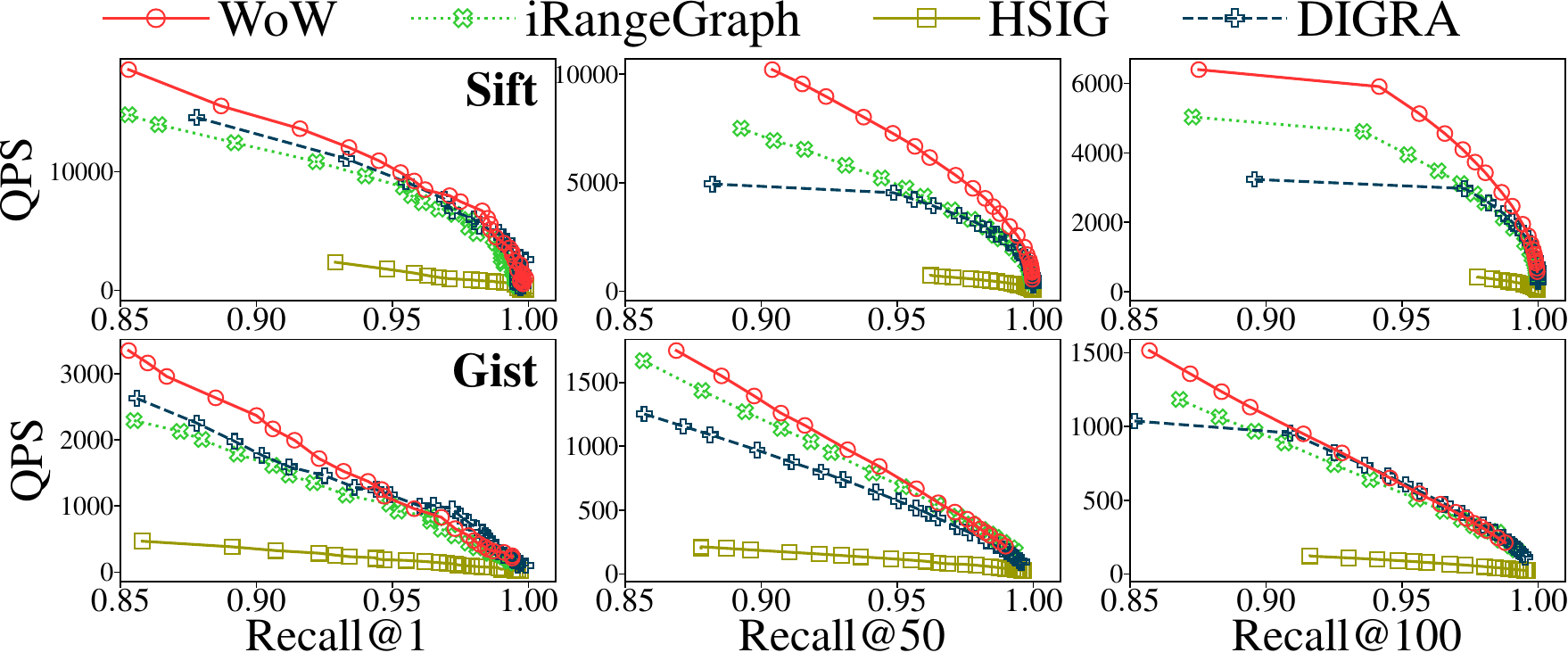}
    \caption{QPS-Recall@$k$ ($k=1,50,100$) on the mixed range fraction workloads of the Sift and Gist datasets}
    \label{fig:varying_k}
    \Description{QPS-Recall@$k$ ($k=1,50,100$) on the mixed range fraction workloads of the Sift and Gist datasets}
\end{figure}

\begin{figure}
    \centering
    \includegraphics[width=0.98\linewidth]{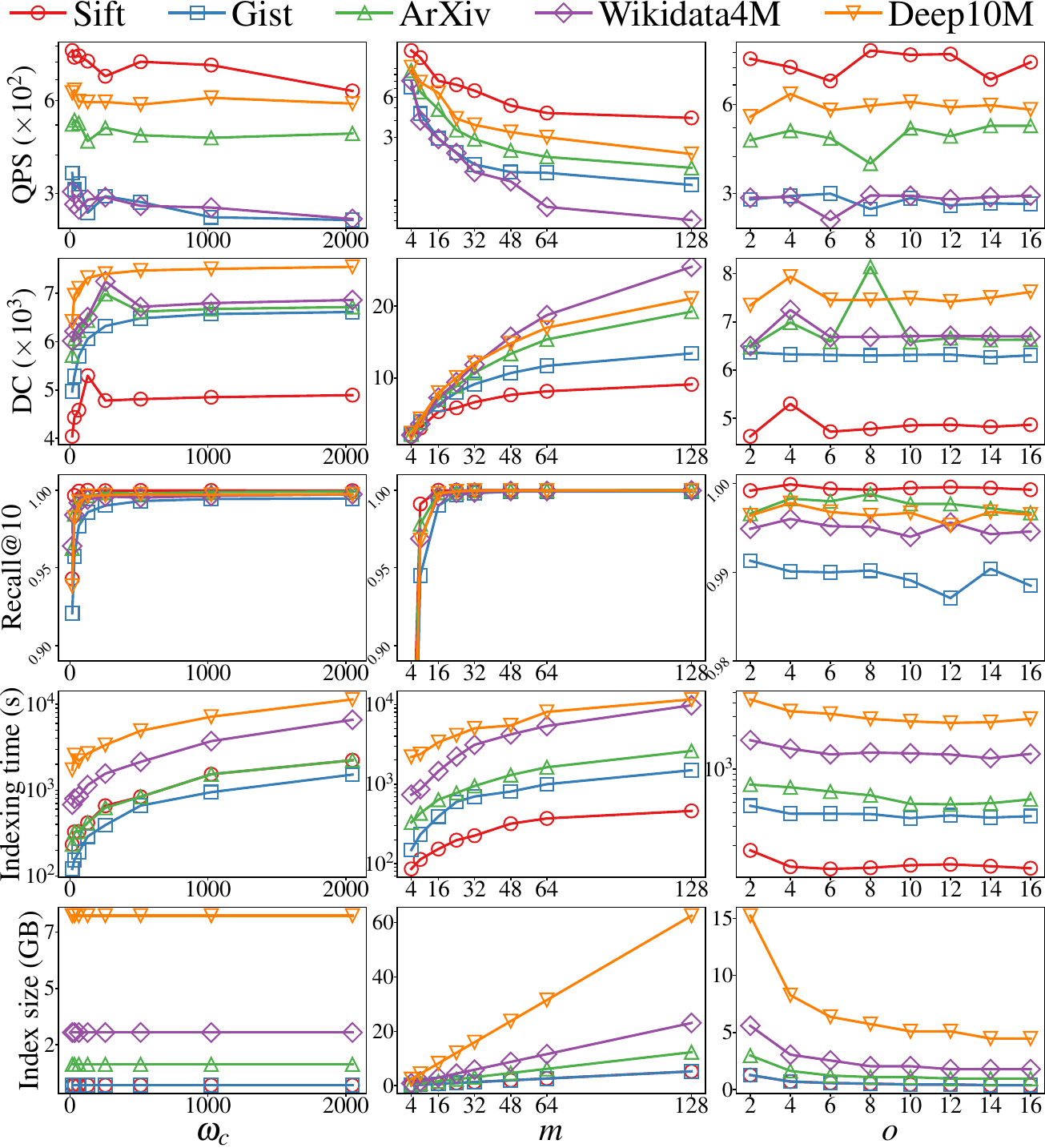}
    \caption{Parameter sensitivity on all datasets}
    \label{fig:param_all}
    \Description{Parameter sensitivity on all datasets}
\end{figure}

\begin{figure*}[t]
    \centering
    \includegraphics[width=\linewidth]{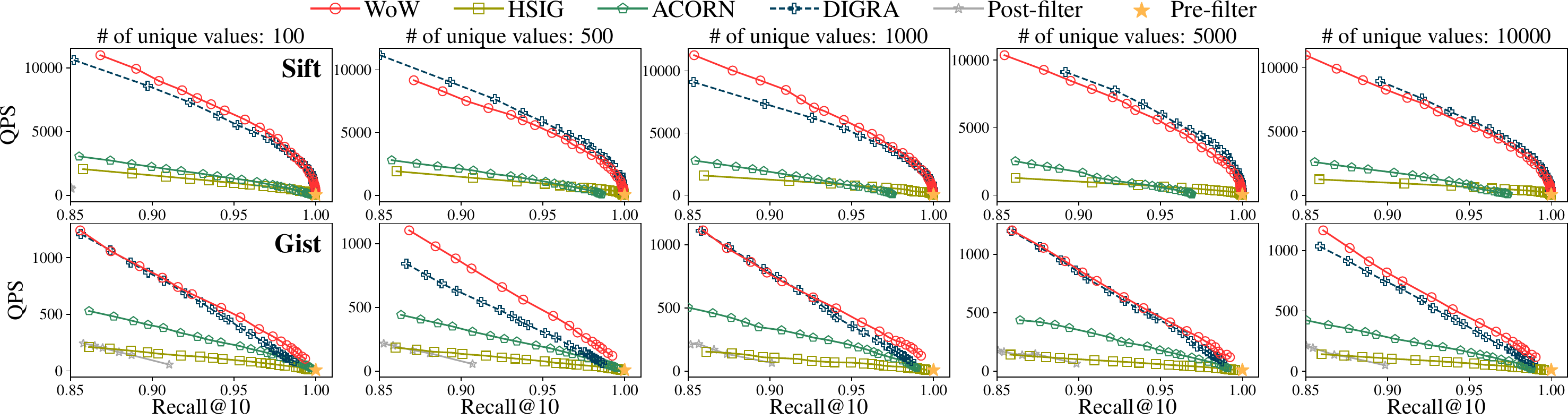}
    \caption{QPS-Recall@10 with varying numbers of unique values. Range filters are generated randomly with mixed selectivity.}
    \label{fig:query-duplicate}
    \Description{QPS-Recall@10 with varying numbers of unique values.}
\end{figure*}

\paragraph{Landing layer selection.}
Apart from low-layer checks optimized by the early-stop strategy, the selectivity-aware landing layer selection in Algorithm~\ref{alg:knn} can also avoid unnecessary high-layer checks.
Figure~\ref{fig:layer_sel} compares QPS of WoW landing on different layers from the \textit{top} layer down to layer 0, on the Gist dataset.
Under all studied workloads, WoW can select the optimal landing layer to achieve the best query performance.
Notably, for vertices that are reachable during graph traversal, a few of them can still pass the filter above the layer selected by Algorithm~\ref{alg:knn}.
However, out-of-range vertices constitute the majority, and are filtered out by a multitude of unnecessary filter checks on the path.
Redundant filter checks may offset the improvement by the reduced distance computations.
For workloads of extreme selectivity ($f=2^{-10}$), the query performance improvement by the layer selection strategy is the most prominent, and the curves of $l_d\in[6,10]$ generally follows the descending order.
WoW can skip most unnecessary filter checks in high layers and directly search in the most relevant layer with the optimal expectation of in-range vertices proved in Theorem~\ref{theo:o_impact}, obtaining the optimal balance between filter checks and distance computations.

\paragraph{Robustness for varying query correlations.}
We evaluate WoW on three workloads with different query correlations.
The correlations are classified by the highest recall of post-filtering: low correlation ($\le 20\%$), no correlation ($40\%$--$60\%$), and high correlation ($\ge 80\%$).
High correlation means most of nearest vectors can satisfy the filter, while low correlation means most of them cannot.
Figure~\ref{fig:query-correlation} shows that WoW has steady performance on different workloads, demonstrating robustness for different query correlations.

\paragraph{Impact of dataset hardness.}
As shown in Table~\ref{tab:data}, LID@10 of Gist is 1.6--2.7 times of Sift.
Figure~\ref{fig:dc-lid} reveals that
\begin{enumerate*}
    \item For a large LID (e.g., 64.86 on Gist range fraction $2^{-1}$), WoW requires over 11,730 DC to achieve 98.03\% recall.
    \item In comparison, for the same fraction $2^{-1}$, Sift has a small LID (31.66), leading to 9,782 DC (16\% fewer than Gist) to achieve 99.91\% recall.
\end{enumerate*}
In short, higher LID can make graph-based indices harder to find results.

\paragraph{Performance for different Recall@$k$.}
    Figure~\ref{fig:varying_k} shows two findings:
    \begin{enumerate*}
        \item As $k$ grows, maintaining the same recall incurs increased query latency.
        \item WoW consistently attains superior query performance across different $k$, validating robustness to retrieval size.
\end{enumerate*}

\paragraph{Parameter sensitivity.}
Figure~\ref{fig:param_all} shows the impact of the hyperparameters, $\omega_c$, $m$, and $o$, on building and searching performance.
Overall, our parameter settings are robust across all datasets in both indexing and query efficiency.
\begin{enumerate}[wide]
\item A greater beam width $\omega_c$ results in more accurate candidates and a more accurate RNG, at the cost of longer indexing time that grows linearly with $\omega_c$.
However, the query speed and accuracy can benefit from the high-quality RNG, as there are more long-distance edges to navigate search to the target neighborhood with fewer DC.
Setting $\omega_c$ from  128 to 256 is sufficient for most datasets.

\item A greater maximum outdegree $m$ leads to a denser RNG and more DC per hop, while a smaller $m$ causes a sparse graph with possibly unreachable vertices.
Setting $m=16$ is sufficient to achieve a high recall for all datasets.

\item As the window boosting base $o$ increases, the layer number, calculated by $ \lceil\log_o\frac{n}{2}\rceil + 1$, decreases.
Fewer layers reduce recall due to a lower fraction of in-range vertices in the landing layer.
E.g., achieving 80\% recall on Wikidata4M needs 521 distance computations when $o=2$, compared to 609 when $o=16$, representing 16\% increase in redundancy.
This result aligns with Theorem~\ref{theo:o_impact}, Case (a).
At $o=4$ and $o=2$, QPS remains competitive. 
$o = 4$ yields higher recall (e.g., 99.83\% vs 99.66\% on ArXiv), slightly more distance computations, and reduced indexing time and index size.
\end{enumerate}

\begin{table}
    \centering
    \caption{Performance of WoW on larger datasets}
    \label{tab:scalability}
    {\small
    \begin{tabular}{lrrrrr}
        \toprule
        & \multicolumn{2}{c}{Indexing} & \multicolumn{3}{c}{QPS-Recall@10} \\
        \cmidrule(lr){2-3} \cmidrule(lr){4-6}
        & Size (MB) & Time (s) & 90\% & 95\% & 99\%\\
        \midrule
        Wikidata4M & 3,112 & 1,557 & 2,518 & 1,486 & 523 \\
        Wikidata41M & 37,666 & 32,183 & 1,212 & 713 & 197\\
        \midrule
        Deep10M & 8,430 & 3,360 & 6,952 & 4,521 & 1,455 \\
        Deep100M & 90,789 & 146,465 & 2,078 & 1,264 & 293 \\
        \bottomrule
    \end{tabular}}
\end{table}


\paragraph{Duplicate attribute values.}
    Section~\ref{sec:extensions} discusses the mechanism to handle duplicate attribute values.
    In this experiment, we evaluate indices with native support for duplicate attribute values.
    Each vector is randomly assigned an attribute value in range $[1, n_c]$, where $n_c$ is the number of unique values.
    Figure~\ref{fig:query-duplicate} shows that WoW and DIGRA have clear advantage on all datasets.
    WoW is still superior to DIGRA for its faster indexing speed, smaller size, parallelism support, and stable performance after insertions.

\paragraph{Scalability to larger datasets}
Table~\ref{tab:scalability} demonstrates the scalability of WoW to larger datasets.
Wikidata41M is the complete vector set embedded from 41,488,110 English entities.
Deep100M consists of 100 million vectors from the original Deep1B dataset.
As the size of Wikidata41M and Deep100M is multiplied by about 10, the index size is multiplied by 10.7 and 12.0, while the indexing time is multiplied by 20.6 and 43.5, resp.
These results are in accord with the index building analysis in Section~\ref{sec:complexity}.
Due to memory limitation, the most competitive indices, iRangeGraph and DIGRA, fail to build on the Deep100M dataset.
Besides, the query speed of WoW decreases sublinearly, consistent with the query time analysis.

\section{Conclusion}
In this paper, we propose a dedicated RFANNS index WoW.
The index addresses two main challenges about incremental construction and efficiency for varying range-filtering query correlations and selectivity.
For incremental construction, WoW leverages a WBT to build hierarchical window graphs without dataset preprocessing.
To handle varying query workloads, WoW employs query selectivity and combines several window graphs to retrieve in-range nearest neighbors.
We prove the time complexity of insertion and query and also suggest optimal parameter settings based on theoretical analysis.
Extensive experiments demonstrate the superiority of WoW in terms of indexing and searching efficiency.
In future work, we plan to study multi-attribute range filtering.
We also want to extend WoW with in-place update and deletion.

\begin{acks}
This work was supported by the Noncommunicable Chronic Diseases - National Science and Technology Major Project (Grant No.: 2023ZD0503600/2023ZD0503604). 
We thank Mr. Shuo Shen for his contributions to deploying baseline methods.
\end{acks}

\bibliographystyle{ACM-Reference-Format}
\bibliography{arxiv}

\newpage
\nobalance
\appendix

\section{Window Calculation} \label{appendix:A}
Algorithm~\ref{alg:get_window} illustrates the window calculation procedure used in Lines 6 and 15 of Algorithm~\ref{alg:insert}, where a window $W_a^l$ with length $2o^l$ halved by $a$ is requested.

\begin{algorithm}[!b]
\SetKw{KwOr}{or}
\SetKw{KwAnd}{and}
\SetKw{KwTrue}{true}
\SetKw{KwFalse}{false}
\SetKw{KwNot}{not}
\SetKw{KwBreak}{break}
\SetKwData{KwHyper}{Hyperparameter}
\SetKwInOut{HyperIn}{Hyperparameter}
\caption{GetWindow}
\label{alg:get_window} 
\KwIn{\textit{a}: attribute value; $l$: window graph layer}
\KwOut{$W_a^l$: window of attribute value $a$ in layer $l$}
\HyperIn{$o$: window boosting base}
$c\leftarrow$ WBT node with value \textit{a}, $r\leftarrow o^l$\;
\If{$r \le |c.\mathcal{T}_{\textit{left}}|$}{
$w_{\min}\leftarrow$ the $r$-th closest value to $a$ in $c.\mathcal{T}_{\textit{left}}$\;
}
\Else{
    $r \leftarrow r - |c.\mathcal{T}_{\textit{left}}|$\;
    \While{$c.\mathcal{T}\neq\emptyset$}{
        $p \leftarrow$ parent of $c$\;
        \lIf{$p.\mathcal{T}=\emptyset$}{
            $w_{\min} \leftarrow$ min value in $p.\mathcal{T}$}
        \ElseIf{$c$ is the root of $p.\mathcal{T}_{\textit{right}}$}{
            \If{$|p.\mathcal{T}_{\textit{left}}|+1\ge r$}{
                \lIf{$r=1$}{$w_{\min} \leftarrow$ value of $p$}
                \lElse{
                    $w_{\min} \leftarrow $ the $(r-1)$-th closest value to $a$ in $p.\mathcal{T}_{\textit{left}}$}}
            \lElse{
                $r \leftarrow r - (|p.\mathcal{T}_{\textit{left}}| + 1)$}}
        \lElse{$c \leftarrow p$}
        \lIf{$w_{\min}$ is determined}{\KwBreak}
        }}
$w_{\max}$ is calculated by a dual version of the above procedure: replacing $\mathcal{T}_{\textit{left}} \leftrightarrow \mathcal{T}_{\textit{right}}$, min $\rightarrow$ max, $w_{\min}\rightarrow w_{\max}$\;
\Return $W_a^l=[w_{\min},w_{\max}]$\;
\end{algorithm}

The algorithm first calculates the left window boundary $w_{\min}$.
For clarity, trees at the children of leaves and the parent of the root are defined as empty.
Two variables are defined in Line 1: $c$ records the currently under-processed node during tree traversal, and $r$ works as a \emph{budget} to record the number of values that should be checked and is greater than the left boundary.
After locating at the node with value $a$ (Line~1), it has to decide which branch to traverse using the size of its left subtree, $|\mathcal{T}_{\textit{left}}|$.
Remember that the tree size rooted at each WBT node is recorded by the node itself.
If $r \le |\mathcal{T}_{\textit{left}}|$ (Line 2), the answer is in the left subtree, and we set $w_{\min}$ to the $r$-th closest value to $a$.
For example, in an ordered array $[1,2,3,4,5]$, the second closest value to 5 is 3 and the forth is 1.
Otherwise, the boundary is less than the smallest value in $\mathcal{T}_{\textit{left}}$, and we need ``climb up'' to find the potential subtree that may contain the left boundary value in Lines 5--15.
In Lines 9--13, we have ``climbed'' from the right child, which means the left subtree of the parent may contain the answer.
If the budget is only one (Line 11), the value of the parent is what we want, otherwise we get the $(r-1)$-th closest value to $a$ in the left subtree (Line 12).
If the left subtree cannot cover the budget (Line 13, $|p.\mathcal{T}_{\textit{left}}|+1<r$), we reduce the budget by $|p.\mathcal{T}_{\textit{left}}|+1$ and keep looking for the potential tree node.
In Line 14, we have ``climbed'' from the left child, and it means that the value of the current node is smaller than the expected left window boundary, thus we keep climbing.
If we have checked all ancestors including the tree root but the window boundary cannot be found (Line 8), $w_{\min}$ is set to the dataset boundary, i.e. the minimum value.

After the determination of $w_{\min}$, $w_{\max}$ can be found with a dual version of the above procedure, which can be generated by replacing $\mathcal{T}_{\textit{left}} \leftrightarrow \mathcal{T}_{\textit{right}}$, min $\rightarrow$ max, $w_{\min}\rightarrow w_{\max}$.
For example, if we have checked all ancestors in Line 8, $w_{\max}$ should be set to the \emph{maximum} (replacing minimum) value in $p.\mathcal{T}$.
The single-branch tree traversal (i.e. only one node is visited on each level of the tree on the search path) is conducted at most three times for one boundary: locating value $a$ in Line 1; ``climbing'' towards root for the potential ancestor; and finding the closest value to $a$ in the potential subtree. 

The time complexity of Algorithm~\ref{alg:get_window} scales to $O(\log n)$, where $n$ denotes the number of tree nodes.

\section{Range Filter Selectivity Calculation} \label{appendix:B}

Algorithm~\ref{alg:get_cardinality} shows how to count the size of the filtered dataset $n'$ used in Line 1 of Algorithm~\ref{alg:knn}, which reflects the selectivity of range filter $R$ in Definition~\ref{def:selectivity}.
The algorithm first calculates the true ranges $[x_u,y_l]$, where the boundaries exist in the dataset and are covered by $R$. 
Then, it calculates their ordered ranks and uses the subtraction as the final cardinality.

The rank calculation procedure starts from the root of WBT and traverses to the node with the target value $a$.
If the target value is less than the value of the visited node, we jump to the left child.
If the target is greater, which means that the rank is also greater, we increase \textit{rank} by $|c.\mathcal{T}_{\textit{left}}| + 1$ and jump to the right child.
Otherwise, the node of the target value is found.
After increasing \textit{rank} by the size of its left subtree where the values are all less than $a$, we return the accumulated \textit{rank} as the result.
The algorithm requires four rounds of single-branch tree traversals.
Thus, the time complexity is also logarithmic.
In the real-world implementation, the complexity can be further improved by combining the boundary determination in Lines 1--2 and the rank calculation in Line 3 together to obtain two round of traversals.

\begin{algorithm}[!b]
\SetAlgoRefName{5}
\SetKw{KwOr}{or}
\SetKw{KwAnd}{and}
\SetKw{KwTrue}{true}
\SetKw{KwFalse}{false}
\SetKw{KwNot}{not}
\SetKw{KwBreak}{break}
\SetKwData{KwHyper}{Hyperparameter}
\SetKwInOut{HyperIn}{Hyperparameter}
\caption{FilteredSetCardinality}
\label{alg:get_cardinality} 
\KwIn{$R=[x,y]$: range filter}
\KwOut{$n'$: cardinality of the filtered dataset}
$x_u\leftarrow$ upper bound of $x$\;
$y_l\leftarrow$ lower bound of $y$\;
$i \leftarrow \mathrm{GetRank}(x_u)$, $j\leftarrow \mathrm{GetRank}(y_l)$\;
\Return $n' = j-i+1$\;
\BlankLine
\SetKwProg{myproc}{procedure}{}{}
\myproc{$\mathrm{GetRank}(a)$}{
$c\leftarrow$ root of WBT, $\textit{rank}\leftarrow 0$\;
\While{$c.\mathcal{T} \neq\emptyset$}{
    \lIf{$a <$ value of $c$}{$c\leftarrow$ root of $c.\mathcal{T}_{\textit{left}}$}
    \ElseIf{$a>$ value of $c$}{
    $\textit{rank}\leftarrow\textit{rank} + |c.\mathcal{T}_{\textit{left}}| + 1$\;
    $c\leftarrow$ root of $c.\mathcal{T_{\textit{right}}}$\;}
    \Else(\Comment*[f]{$a=$ value of $c$, target found.}){
        $\textit{rank}\leftarrow \textit{rank} + |c.\mathcal{T}_{\textit{left}}|$\; 
        \KwBreak\;
    }
}
\Return \textit{rank}\;
}
\end{algorithm}

\section{Proof of Theorem~3.2} \label{appendix:C}
\renewcommand{\thetheorem}{3.2}
\begin{theorem} \label{theo:o_impact_appendix}
    In the landing layer $l_d$, the expected fraction $f_R$ of in-range neighbors at a single hop on the path is bounded by
    
    \begin{equation}
        f_R=\begin{cases}
            (\frac{1}{\sqrt{o}},\frac{1}{2}) & \text{(a) } l\in (l'-1,l'-\frac{1}{2}), o>4, \\
            \left[\frac{\sqrt{2}}{2}-\frac{1}{4o^{l+1}},\frac{3}{4}-\frac{1}{4o^{l+1}}\right) &  \text{(b) } l\in (l'-1,l'-\frac{1}{2}), o \leq 4, \\
            \left[\frac{3}{4}-\frac{1}{4o^{l}}, 1-\frac{o^l + 1}{4o^{l + \frac{1}{2}}}\right] &  \text{(c) } l \in [l'-\frac{1}{2},l'],
        \end{cases}
    \end{equation}
    where $n'$ is the number of in-range vectors for filter $R=[x,y]$, $l'=\log_o\frac{n'}{2}$, and $l = l_h = \lfloor\log_o\frac{n'}{2}\rfloor$, which is the highest layer defined in Line~2 of Algorithm~3, whose window size is less than $R$.
\end{theorem}
\renewcommand{\thetheorem}{\thesection.\arabic{theorem}}
\begin{proof}
    The landing layer $l_d$ is determined by comparing $\frac{2o^l}{n'}$ and $\frac{n'}{2o^{l+1}}$.
    Range filter $[x,y]$ can be rewritten to $[x, x+n'-1]$ by assuming all attribute values are sequential.
    Moreover, vectors have an equal probability to be selected as neighbors in a certain range.
    There are two situations:

    $\bullet$ \textbf{Situation 1.} $2o^{l + \frac{1}{2}} < n' < 2o^{l + 1}$, then $\frac{2o^l}{n'} < \frac{n'}{2o^{l+1}}$ and $l_d = l + 1$, the half window size is $o^{l+1}$.  For a vector with attribute value $x + i$ in range $[x, x + n' - 1]$ where $i \in [0, n'-1]$, the window of a vector with attribute value $x + i$ is $W_{x + i} ^{l+1} = [x + i-o^{l+1} + 1, x + i + o ^ {l+1} - 1]$.
        
    \underline{Case (a).}
    When $n' < o^{l+1}$, the range filter is always covered by the windows.
    In this case, $2o^{l+\frac{1}{2}} < o^{l+1}$ and $o > 4$.
    The fraction can be calculated by $f_R=\frac{n'}{2o^{l+1}}\in(\frac{1}{\sqrt{o}}, \frac{1}{2})$.

    \underline{Case (b).}
    When $n' \geq o^{l+1}$, some windows fail to cover the whole range. 
    If the window left boundary is less than the filter left boundary, i.e. $x+i-o^{l+1} < x \Rightarrow i\le o^{l+1}-1$
    The fraction of in-range vectors is $f_1= \frac{(x + i + o^{l+1} - 1 )- x + 1}{2o^{l+1}}$, and the expectation is
    \begin{equation}
    \bar{f_1}=\mathbb{E}_{i\in[0,o^{l+1}-1]}\left[\frac{i + o^{l+1}}{2o^{l+1}}\right],
    \end{equation}
    which is $\bar{f_1}=\frac{1}{4}(3-o^{l+1})$.
    On the other hand, if $i\in [o^{l+1}, n'-1]$, the window right boundary is greater than the filter right boundary.
    The fraction is $f_2=\frac{(x+n'-1) - (x+i-o^{l+1}+1) + 1}{2o^{l+1}}$, and
    \begin{equation}
    \bar{f_2}=\mathbb{E}_{i\in[o^{l+1},n'-1]}\left[\frac{n'+o^{l+1}-1-i}{2o^{l+1}}\right],
    \end{equation}
    which is $\bar{f_2}=\frac{n'+o^{l+1}-1}{4o^{l+1}}$.
    Combining two parts, $f_R=w_1\bar{f_1}+w_2\bar{f_2}$, where $w_1=\frac{o^{l+1}}{n'}$ and $w_2=\frac{n'-o^{l+1}}{n'}$,
    the total fraction is
    
    \begin{equation}
    \begin{aligned}
        f_R&=\frac{3o^{l+1}-1}{4o^{l+1}}\cdot\frac{o^{l+1}}{n'} + \frac{n'+o^{l+1}-1}{4o^{l+1}}\cdot\frac{n'-o^{l+1}}{n'}\\
        &=\frac{o^{l+1}}{2n'}+\frac{n'-1}{4o^{l+1}}\in\left[\frac{1}{4}\left(2\sqrt{2}-o^{-(l+1)}\right),\frac{1}{4}\left(3-o^{-(l+1)}\right)\right),
        \end{aligned}
    \end{equation}
    and the lower bound is at $ n'=\sqrt{2}o^{l+1}\in[2o^{l+\frac{1}{2}},2o^{l+1})$, while the upper bound is determined by $n'=2o^{l+1}$.
    
    $\bullet$ \textbf{Situation 2.} \underline{Case (c).} $ 2o^l\le n' \le 2o^{l+\frac{1}{2}}$, then $\frac{2o^l}{n'} \ge \frac{n'}{2o^{l+1}}$ and $l_d = l$, the half window size is $o^l$.
    For a vector with attribute value $x + i$ in range $[x, x + n' - 1]$ where $i \in [0, n'-1]$, its window in Layer $l$ is $W_{x+i}^l = [x + i-o^l + 1, x + i + o ^l - 1]$. 
    
    When $i \le o^l - 1$, we have $x +i-o^l + 1 < x$ that the left boundary of the window is less than that of the filter. 
    The fraction of in-range vectors over the window size is $f_1= \frac{(x + i + o^l - 1)-x+1}{2o^l}$.
    The expected fraction of this situation is
    \begin{equation}
        \bar{f_1} = \mathbb{E}_{i\in[0,o^l-1]}\left[\frac{i + o^l}{2o^l}\right],
    \end{equation}
    which is $\bar{f_1} = \frac{1}{4}(3-o^{-l})$.
    When $i \ge n'-o^l$, we have $x + n'-1 < x + i + o^l - 1$ that the right boundary of the filter is less than that of the window. The fraction $f_2 = \frac{(x + n'-1) - (x+i-o^l + 1) + 1}{2o^l}$.
    We can calculate $\bar{f_2}=\frac{1}{4}(3-o^{-l})$ in the same way as $\bar{f_1}$.
    When $o^l-1 <i< n'-o^l$, the expected fraction is $\bar{f_3} = 1$ because all windows of these elements are covered by the filter.

    Finally, the total fraction of in-range neighbors is $w_1 \bar{f_1} + w_2\bar{f_2} + w_3\bar{f_3}$ where $w_1 = w_2 = \frac{o^l}{n'}$ and $w_3 = \frac{n'-2o^l}{n'}$.
    \begin{equation}
    \begin{aligned}
        f_R & = \frac{2o^l}{n'}\left(\frac{1}{4}\left(3-o^{-l}\right)\right) + \frac{n'-2o^l}{n'} \cdot 1 \\
        & =1-\frac{o^l + 1}{2n'} \in \left[\frac{1}{4}\left(3-o^{-l}\right), 1-\frac{o^l + 1}{4o^{l + \frac{1}{2}}}\right].
    \end{aligned}
    \end{equation}
\end{proof}

\end{document}